\documentclass[runningheads]{llncs}

\usepackage[utf8]{inputenc}
\usepackage[english]{babel}
\usepackage{authblk}
\usepackage[section]{placeins}%so that figures appear on the right section
\usepackage{graphicx}
\usepackage{amssymb,amsmath}
\usepackage{MnSymbol}
\usepackage{mathrsfs}
\usepackage{delarray}
\usepackage{stmaryrd}
\usepackage{ifthen}
\usepackage{paralist}
\usepackage[table]{xcolor}
\usepackage{diagbox}
\usepackage{xcolor}
\usepackage{color}%change color of single cell of table
\usepackage{multirow}
\usepackage{multicol}%multiple columns of text
\usepackage{makecell}%multiple rows in one cell
\usepackage{tikz}
\usepackage{hyperref}
\graphicspath{{figures/}}

\newtheorem{remark*}{Remark}
\newtheorem{thm}{Theorem}
\newtheorem{cor}{Corollary}

\newcommand{\N}{\mathbb{N}}
\newcounter{finaln}
\newcommand{\integers}[2][2]{\ifthenelse{\equal{#1}{2}}
{\llbracket #2\rrbracket}
{\ifthenelse{\equal{#1}{0}}{
\ifthenelse{\equal{11}{\the\catcode`#2}}
{\{0,\ldots,#2 -1\}}
{\setcounter{finaln}{#2 -1}
\{0,\ldots,\thefinaln \}}}
{\{1,\ldots,#2\}}}}

\renewcommand{\int}[1]{\integers{#1}}
\newcommand{\lcm}{\text{lcm}}

\newcommand{\T}{\mathscr{T}} % trajectory
\renewcommand{\O}{\mathscr{O}} % orbit
\newcommand{\C}{\mathscr{C}} % limit dynamics
\newcommand{\G}{\mathscr{G}} % transition graph
\DeclareMathOperator{\deltaeq}{\;\overset{\Delta}{=}\;}

\newcommand{\p}{\textsc{par}}
\newcommand{\bip}{\textsc{bip}}
\newcommand{\seq}{\textsc{seq}}
\newcommand{\bs}{\textsc{bs}}
\newcommand{\bp}{\textsc{bp}}
\newcommand{\lc}{\textsc{lc}}

\newcommand{\PreserveBackslash}[1]{\let\temp=\\#1\let\\=\temp}
\newcolumntype{C}[1]{>{\PreserveBackslash\centering}p{#1}}
\newcommand{\Z}{\mathbb{Z}}

\newcommand{\B}{\mathbb{B}}
\newcommand{\Bn}{\mathbb{B}^n}

\renewcommand{\O}{\mathcal{O}}

\newcommand{\BS}[1]{{\text{\textsc{BS}}_{#1}}}
\newcommand{\even}{\mathrm{e}}
\newcommand{\odd}{\mathrm{o}}
\newcommand{\barmu}{\overline{\mu}}

\newcommand{\ie}{i.e.{}}

\title{
  On elementary cellular automata asymptotic (a)synchronism sensitivity 
  and complexity 
}

\titlerunning{
	ECA asymptotic (a)synchronism sensitivity and complexity
}

\author{
  Isabel Donoso Leiva\inst{1,3} \and
  Eric Goles\inst{1} \and
  Mart{\'i}n R{\'i}os-Wilson\inst{1} \and
  Sylvain Sen{\'e}\inst{2,3}
}

\authorrunning{
	I. Donoso Leiva, E. Goles, M. R{\'i}os Wilson, and S. Sen{\'e}
}

\institute{
  Facultad de Ingenier{\'i}a y Ciencias, Universidad Adolfo 
  Ib{\'a}{\~n}ez \and
  Universit{\'e} publique, Marseille, France \and
  Aix Marseille Univ, CNRS, LIS, Marseille, France
}

\begin{document}
\maketitle

%%%%%%%%%%%%%%%%%%%%%%%%%%
\begin{abstract}
  Among the fundamental questions in computer science is that of the impact of 
  synchronism/asynchronism on computations, which has been addressed in 
  various fields of the discipline: in programming, in networking, in 
  concurrence theory, in artificial learning, etc. 
  In this paper, we tackle this question from a standpoint which mixes 
  discrete dynamical system theory and computational complexity, by
  highlighting that the chosen way of making local computations can have a 
  drastic influence on the performed global computation itself.
  To do so, we study how distinct update schedules may fundamentally change 
  the asymptotic behaviors of finite dynamical systems, by analyzing in 
  particular their limit cycle maximal period. 
  For the message itself to be general and impacting enough, we choose to 
  focus on a ``simple'' computational model which prevents underlying systems 
  from having too many intrinsic degrees of freedom, namely elementary 
  cellular automata.
  More precisely, for elementary cellular automata rules which are neither too 
  simple nor too complex (the problem should be meaningless for both), we show 
  that update schedule changes can lead to significant computational 
  complexity jumps (from constant to superpolynomial ones) in terms of their 
  temporal asymptotes.
\end{abstract}

%%%%%%%%%%%%%%%%%%%%%%%%%%
\section{Introduction}

In the domain of discrete dynamical systems at the interface with computer 
science, the generic model of automata networks, initially introduced in 
the 1940s through the seminal works of McCulloch and Pitts on formal neural 
networks~\cite{McCulloch1943} and of Ulam and von Neumann on cellular 
automata~\cite{vonNeumann1966}, has paved the way for numerous fundamental 
developments and results; for instance with the introduction of finite 
automata~\cite{Kleene1956}, the retroaction cycle theorem~\cite{Robert1980}, 
the undecidability of all nontrivial properties of limit sets of cellular 
automata~\cite{Kari1994}, the Turing universality of the model 
itself~\cite{Smith1971,Goles1990}, ... and also with the generalized use of 
Boolean networks as a representational model of biological regulation networks 
since the works of Kauffman~\cite{Kauffman1969} and Thomas~\cite{Thomas1973}.
Informally speaking, automata networks are collections of discrete-state 
entities (the automata) interacting locally with each other over discrete time 
which are simple to define at the static level but whose global dynamical 
behaviors offer very interesting intricacies.

Despite major theoretical contributions having provided since the 1980s a 
better comprehension of these objects~\cite{Robert1986,Goles1990} from 
computational and behavioral standpoints, understanding their sensitivity to 
(a)synchronism remains an open question on which any advance could have deep 
implications in computer science (around the thematics of synchronous versus 
asynchronous computation and processing~\cite{Chapiro1984,Charron1996}) and in 
systems biology (around the temporal organization of genetic 
expression~\cite{Hubner2010,Fierz2019}). 
In this context, numerous studies have been published by considering distinct 
settings of the concept of synchronism/asynchronism, i.e. by defining update 
modes which govern the way automata update their state over time.
For instance, (a)synchronism sensitivity has been studied \emph{per se} 
according to deterministic and non-deterministic semantics 
in~\cite{Goles2008,Aracena2011,Noual2018} for Boolean automata networks and 
in~\cite{Ingerson1984,Fates2005,Dennunzio2013,Dennunzio2017} for cellular 
automata subject to stochastic semantics.

In these lines, the aim of this paper is to increase the knowledge on Boolean 
automata networks and cellular automata (a)synchronism sensitivity. 
To do so, we choose to focus on the impact of different kinds of periodic 
update modes on the dynamics of elementary cellular automata: from the most 
classical parallel one to the more general local clocks 
one~\cite{Pauleve2022}. 
Because we want to exhibit the very power of update modes on dynamical systems 
and concentrate on it, the choice of elementary cellular automata is quite 
natural: 
they constitute a restricted and ``simple'' cellular automata family which is 
well known to have more or less complex representatives in terms of dynamical 
behaviors ~\cite{Wolfram1984,Culik1988,Kurka1997} without being too much 
permissive.
For our study, we use an approach derived from~\cite{Rios2021a} and we 
pay attention to the influence of update modes on the asymptotic dynamical 
behaviors they can lead to compute, in particular in terms of limit cycles 
maximal periods.

In this paper, we highlight formally that the choice of the update mode can 
have a deep influence on the dynamics of systems. 
In particular, two specific elementary cellular automata rules, namely rules 
$156$ and $178$ (a variation of the well-known majority function tie case) as 
defined by the Wolfram's codification, are studied here. 
They have been chosen from the experimental classification presented 
in~\cite{Wolfram1984}, as a result of numerical simulations which have given 
the insight that they are perfect representatives for highlighting 
(a)synchronism sensitivity.
Notably, they all belong to the Wolfram's class II, which means that, 
according to computational observations, these cellular automata evolved 
asymptotically towards a ``set of separated simple stable or periodic 
structures''.
Since our (a)synchronism sensitivity measure consists in limit cycles maximal 
periods, this  Wolfram's class II is naturally the most pertinent one in our 
context. 
Indeed, class I cellular automata converge to homogeneous fixed points, class 
III cellular automata leads to aperiodic or chaotic patterns, and class IV 
cellular automata, which are deeply interesting from the computational 
standpoint, are not relevant for our concern because of their global high 
expressiveness which would prevent from showing asymptotic complexity jumps 
depending on update modes. 
On this basis, for these two rules, we show between which kinds of update 
modes asymptotic complexity changes appear.
What stands out is that each of these rules admits it own (a)synchronism 
sensitivity scheme (which one could call its own asymptotic complexity scheme 
with respect to synchronism), which supports interestingly the existence of a 
periodic update modes expressiveness hierarchy.\smallskip

In Section~\ref{sec:def}, the main definitions and notations are formalized. 
The emphasizing of elementary cellular automata (a)synchronism sensitivity 
is presented in Section~\ref{sec:res} through upper-bounds for the limit-cycle 
periods of rules $156$ and $178$ depending on distinct families of periodic 
update modes. 
The paper ends with Section~\ref{sec:persp} in which we discuss some 
perspectives of this work.
Full proofs can be found in the long version of this paper available \href{}
{here}.

%%%%%%%%%%%%%%%%%%%%%%%%%%
\section{Definitions and notations}
\label{sec:def}

\paragraph{General notations}

Let $\integers{n} = \integers[0]{n}$,
let $\B = \{0,1\}$, and 
let $x_i$ denote the $i$-th component of vector $x \in \B^n$.
Given a vector $x \in \B^n$, we can denote it classically as 
$(x_0, \dots, x_{n-1})$ or as the word $x_0 \dots x_{n-1}$ if it eases the 
reading.

\subsection{Boolean automata networks and elementary cellular automata}

Roughly speaking, a Boolean automata network (BAN)  applied over a grid of 
size $n$ is a collection of $n$ automata represented by the set 
$\integers{n}$, each having a state within $\B$, which interact with each 
other over discrete time.  
A \emph{configuration} $x$ is an element of $\Bn$, \ie a Boolean vector of 
dimension $n$. 
Formally, a \emph{BAN} is a function $f: \Bn \to \Bn$ defined by means of $n$ 
local functions $f_i: \Bn \to \B$, with $i \in \integers{n}$, such that $f_i$ 
is the $i$th component of $f$. 
Given an automaton $i \in \integers{n}$ and a configuration $x \in \Bn$, 
$f_i(x)$ defines the way that $i$ updates its states depending on the state of 
automata \emph{effectively} acting on it; automaton $j$ ``effectively'' acts 
on $i$ if and only if there exists a configuration $x$ in which the state of 
$i$ changes with respect to the change of the state of $j$; $j$ is then called 
a \emph{neighbor} of $i$.\smallskip
 
An \emph{elementary cellular automaton (ECA)} is a particular BAN dived into 
the cellular space $\Z$ so that \emph{(i)} the evolution of state $x_i$ of 
automaton $i$ (rather called cell $i$ in this context) over time only depends 
on that of cells $i-1$, $i$ itself, and $i+1$, and \emph{(ii)} all cells share 
the same and unique local function.
As a consequence, it is easy to derive that there exist $\smash{2^{2^3}} = 
256$ distinct ECA, and it is well known that these ECA can be grouped into 
$88$ equivalence classes up to symmetry.\smallskip

In absolute terms, BANs as well as ECA can be studied as infinite models of 
computation, as it is classically done in particular with ECA.
In this paper, we choose to focus on finite ECA, which are ECA whose 
underlying structure can be viewed as a torus of dimension $1$ which leads 
naturally to work on $\Z/n\Z$, the ring of integers modulo $n$ so that the 
neighborhood of cell $0$ is $\{n-1, 0, 1\}$ and that of cell $n-1$ is 
$\{n-2, n-1, 0\}$.\smallskip

Now the mathematical objects at stake in this paper are statically defined, 
let us specify how they evolve over time, which requires defining when the 
cells state update, by executing the local functions. 

\begin{figure}[t!]
	\centerline{
		\scalebox{.95}{
		  \includegraphics[page=1]{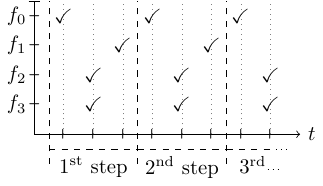}
		}
		\quad
		\scalebox{.95}{
		  \includegraphics[page=2]{figures.pdf}
		}
	}\medskip
	
	\centerline{
		\scalebox{.95}{
		  \includegraphics[page=3]{figures.pdf}		
		}
	}
	\caption{Illustration of the execution over time of local transition 
		functions of any BAN $f$ of size $4$ according to 
		(top left) $\mu_\bs = (\{0\}, \{2,3\}, \{1\})$,
		(top right center) $\mu_\bp = \{(1), (2,0,3)\}$, and
		(bottom) $\mu_\lc = ((1,3,2,2), (0,2,1,0))$.
		The $\checkmark$ symbols indicate the moments at which the automata update 
		their states; the vertical dashed lines separate periodical time steps 
		from each other.}
	\label{fig:um_exec}
\end{figure}

\subsection{Update modes}

To choose an organization of when cells update their state over time leads to 
define what is classically called an update mode (aka update schedule or 
scheme).
In order to increase our knowledge on (a)synchronism sensitivity, as evoked in 
the introduction, we pay attention in this article to deterministic and 
periodic update modes.
Generally speaking, given a BAN $f$  applied over a grid of size $n$, a 
\emph{deterministic} (resp. \emph{periodic}) \emph{update mode} of $f$ is an 
infinite (resp. a finite) sequence $\mu = (B_k)_{k \in \N}$ (resp. $\mu = 
(B_0, \dots, B_{p-1})$), where $B_i$ is a subset of $\integers{n}$ for all $i 
\in \N$ (resp. for all $i \in \integers{p}$). 
Another way of seeing the update mode $\mu$ is to consider it as a function 
$\mu^\star: \N \to \powerset(\integers{n})$ which associates each time step 
with a subset of $\integers{n}$ so that $\mu^\star(t)$ gives the automata 
which update their state at step $t$; furthermore, when $\mu$ is periodic, 
there exists $p \in \N$ such that for all $t \in \N$, $\mu^\star(t+p) = 
\mu^\star(t)$.\smallskip

Three known update mode families are considered: 
the block-sequential~\cite{Robert1986}, 
the block-parallel~\cite{Demongeot2020,Perrot2023} and 
the local clocks~\cite{Rios2021-phd} ones.
Updates induced by each of them over time are depicted in 
Figure~\ref{fig:um_exec}.\smallskip

A \emph{block-sequential update mode} $\mu_\bs = (B_0, \dots, B_{p-1})$ is an 
ordered partition of $\integers{n}$, with $B_i$ a subset of $\integers{n}$ for 
all $i$ in $\integers{p}$. 
Informally, $\mu_\bs$ defines an update mode of period $p$ separating 
$\integers{n}$ into $p$ disjoint blocks so that all automata of a same block 
update their state in parallel while the blocks are iterated in series. 
The other way of considering $\mu_\bs$ is: 
$\forall t \in \N, \mu_\bs^\star(t) = B_{t \mod p}$.\smallskip

A \emph{block-parallel update mode} $\mu_\bp = \{S_0, \dots, S_{s-1}\}$ is a 
partitioned order of $\integers{n}$, with $S_j = (i_{j,k})_{0 \leq k \leq 
|S_j|-1}$ a sequence of $\integers{n}$ for all $j$ in $\integers{s}$. 
Informally, $\mu_\bp$ separates $\integers{n}$ into $s$ disjoint subsequences 
so that all automata of a same subsequence update their state in series while 
the subsequences are iterated in parallel.
Note that there exists a natural way to convert $\mu_\bp$ into a sequence of 
blocks of period $p = \lcm(|S_0|, \dots, |S_{s-1}|)$.
It suffices to define function $\varphi$ as: 
$\varphi(\mu_\bp) = (B_\ell)_{\ell \in \integers{p}}$ with $B_\ell = 
\{i_{j,\ell \mod |S_j|} \mid j \in \integers{s}\}$.  
The other way of considering $\mu_\bp$ is:
$\forall j \in \integers{s}, \forall k \in \integers{|S_j|}$, $i_{j,k} \in 
\mu_\bp^\star(t) \iff k = t \mod |S_j|$.\smallskip

A \emph{local clocks update mode} $\mu_{\lc} = (P,\Delta)$, with 
$P = (p_0, \dots, p_{n-1})$ and $\Delta = (\delta_0, \dots, \delta_{n-1})$, is 
an update mode such that each automaton $i$ of $\integers{n}$ is associated 
with a period $p_i \in \N^*$ and an initial shift $\delta_i \in 
\integers{p_i}$ such that $i \in \mu_{\lc}^\star(t) \iff t = \delta_i \mod 
p_i$, with $t \in \N$.\smallskip

Let us now introduce three particular cases or subfamilies of these three 
latter update mode families. 
The \emph{parallel update mode} $\mu_\p = (\integers{n})$ makes every 
automaton update its state at each time step, such that $\forall t \in \N, 
\mu_\p^\star(t) = \integers{n}$.
A \emph{bipartite update mode} $\mu_\bip = (B_0, B_1)$ is a block-sequential 
update mode composed of two blocks such that the automata in a same block do 
not act on each other (in our ECA framework, this definition induces that such 
update modes are associated to 1D tori of even size and that there are two 
such update modes).
A \emph{sequential update mode} $\mu_\seq = (\phi(\integers{n}))$, where $
\phi(\integers{n}) = \{i_0\}, \dots, \{i_{n-1}\}$ is a permutation of 
$\integers{n}$, makes one and only one automaton update its state at each time 
step so that all automata have updated their state after $n$ time steps 
depending on the order induced by $\phi$.
All these update modes follow the order of inclusion pictured in 
Figure~\ref{fig:um_inclusion}.\smallskip

As we focus on periodic update modes, let us differentiate two kinds of time 
steps.
A \emph{substep} is a time step at which a subset of automata change their 
states.
A \emph{step} is the composition of substeps having occurred over a period.

\begin{figure}[t!]
	\centerline{
		\scalebox{.95}{
		  \includegraphics[page=4]{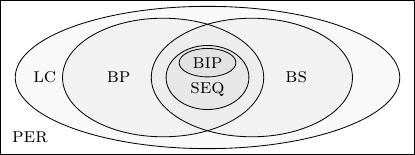}
		}
	}
	\caption{Order of inclusion of the defined families of periodic update 
		modes, where \textsc{per} stands for ``periodic''.}
	\label{fig:um_inclusion}
\end{figure}

\subsection{Dynamical systems}

An ECA $f \in \integers{256}$ together with an update mode $\mu$ define a 
\emph{discrete dynamical system} denoted by the pair $(f, \mu)$. 
$(f, M)$ denotes by extension any dynamical system related to $f$ under the 
considered update mode families, with $M \in \{\textsc{par}, \textsc{bip}, 
\textsc{seq}, \textsc{bp}, \textsc{bs}, \textsc{lc}, 
\textsc{per}\}$.\smallskip

Let $f$ be an ECA of size $n$ and let $\mu$ be a periodic update 
mode represented as a periodical sequence of subsets of $\integers{n}$ such 
that$\mu = (B_0, ..., B_{p-1})$.
Let $F = (f, \mu)$ be the global function from $\B^n$ to itself which defines 
the dynamical system related to ECA $f$ and update mode $\mu$. 
Let $x \in \B^n$ a configuration of $F$.\\
The \emph{trajectory} of $x$ is the infinite path 
$\T(x) \deltaeq x^0 = x \to x^1 = F(x) \to \dots \to x^t = F^t(x) \to \cdots$, 
where $t \in \N$ and 
\begin{multline*}
	F(x) = f_{B_{p-1}} \circ \dots \circ f_{B_0}\text{,}\\
	\text{where } \forall k \in \integers{p}, \forall i \in \integers{n}, 
	f_{B_k}(x)_i = \begin{cases}
			f_i(x) & \text{if } i \in B_k\text{,}\\
			x_i & \text{otherwise,}
	\end{cases}\\
	\text{and } F^t(x) = \underbrace{F \circ \dots \circ F}_{t \text{ times}}(x)
	\text{.}
\end{multline*}

\begin{figure}[t!]
	\centerline{
		\begin{minipage}{.25\textwidth}
			\centerline{
  			\scalebox{.95}{
	  		  \includegraphics[page=5]{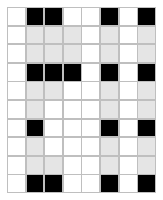}
        }
      }
		\end{minipage}
		\quad
		\begin{minipage}{.25\textwidth}
			\centerline{
			  \scalebox{.95}{
  			  \includegraphics[page=6]{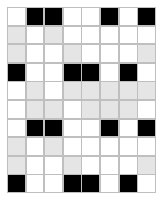}
  			}
  		}
		\end{minipage}
	}
	\caption{Space-time diagrams (time going downward) representing the $3$ 
  	first (periodical) steps of the evolution of configuration 
		$x = (0,1,1,0,0,1,0,1)$ of dynamical systems 
		(left) $(156, \mu_\bs)$, and
		(right) $(178, \mu_\bp)$, %and (right) $(184, \mu_\lc)$, 
    where $\mu_\bs = (\{1,3,4\}, \{0,2,6\}, \{5,7\})$, and
		$\mu_\bp = \{(1,3,4), (0,2,6), (5), (7)\}$, 
		%$\mu_\lc = ((3,1,3,3,1,1,1,3),(0,1,2,1,1,2,0,2))$. 
		The configurations obtained at each step are depicted by lines with cells 
		at state $1$ in black. 
		Lines with cells at state $1$ in light gray represent the configurations 
		obtained at substeps.
		Remark that $x$ belongs to a limit cycle of length $3$ (resp. $2$) in 
		$(156, \mu_\bs)$ (resp. $(178, \mu_\bp)$).}
	\label{fig:space-time-diag}
\end{figure}

In the context of ECA, it is convenient to represent trajectories by 
\emph{space-time diagrams} which give a visual aspect of the latter, as 
illustrated in Figure~\ref{fig:space-time-diag}. 
The \emph{orbit} of $x$ is the set $\O(x)$ composed of all the configurations 
which belongs to $\T(x)$.
Since $f$ is defined over a grid of finite size and the boundary condition is 
periodic, the temporal evolution of $x$ governed by the successive 
applications of $F$ leads it undoubtedly to enter into a \emph{limit 
phase}, \ie a cyclic subpath $\C(x)$ of $\T(x)$ such that $\forall y = 
F^k(x) \in \C(x), \exists t \in \N, F^t(y) = y$, with $k \in \N$.
$\T(x)$ is this separated into two phases, the limit phase and the 
\emph{transient phase} which corresponds to the finite subpath 
$x \to \dots \to x^\ell$ of length $\ell$ such that $\forall i \in 
\integers{\ell + 1}, \nexists t \in \N, x^{i + t} = x^{i}$. 
The \emph{limit set} of $x$ is the set of configurations belonging to 
$\C(x)$.\smallskip

From these definitions, we derive that $F$ can be represented as a graph 
$\G_F = (\B^n, T)$, where $(x,y) \in T \subseteq \B^n \times \B^n \iff 
y = F(x)$.
In this graph, which is classically called a \emph{transition graph}, 
the non-cyclic (resp. cyclic) paths represent the transient (resp. limit) 
phases of $F$. 
More precisely, the cycles of $\G_F$ are the \emph{limit cycles} of $F$. 
When a limit cycle is of length $1$, we call it a \emph{fixed point}. 
Furthermore, if the fixed point is such that all the cells of the 
configuration has the same state, then we call it an \emph{homogeneous fixed 
point}.\smallskip

Eventually, we make use of the following specific notations.
Let $x \in \B^n$ be a configuration and $[i,j] \subseteq \integers{n}$ be a 
subset of cells. 
We denote by $x_{[i,j]}$ the projection of $x$ on $[i,j]$. 
Since we work on ECA over tori, such a projection defines a sub-configurations 
and can be of three kinds: either $i < j$ and $x_{[i,j]} = (x_i, x_{i+1}, 
\dots, x_{j-1}, x_{j})$, or $i = j$ and $x_{[i,j]} = (x_i)$, or $i > j$ and 
$x_{[i,j]} = (x_i, x_{i+1}, \dots, x_n, x_0, \dots, x_{j-1}, x_j$).
Thus, given $x \in \B^n$ and $i \in \integers{n}$, an ECA $f$ can be rewritten 
as
\begin{equation*}
  f(x) = (f(x_{[n-1,1]}), f(x_{[0,2]}), \dots, f(x_{[i,i+2]}), \dots, 
  f(x_{[n-3,n-1]}), f(x_{[n-2,0]})\text{.}
\end{equation*}

\noindent Abusing notations, the word $u \in \mathbb{B}^k$ is called a 
\emph{wall} for a dynamical system if for all $a, b \in \B$, $f(aub) = u$, and 
we assume in this work that walls are of size $2$, i.e. $k = 2$, unless 
otherwise stated. 
Such a word $u$ is an \emph{absolute wall} (resp. a \emph{relative wall}) for 
an ECA rule if it is a wall for any update mode (resp. strict subset of update 
modes).
We say that a rule $F$ can \emph{dynamically create new walls} if there is a 
time $t\in\N$ and an initial configuration $x^0\in\B^n$ such that 
$x^t(=F^t(x^0))$ has a higher number of walls than $x^0$. 
Finally, we say that a configuration $x$ is an \emph{isle} of $1$s (resp. an 
isle of $0$s) if there exists an interval $I = [a,b] \subseteq \integers{n}$ 
such that $x_i = 1$ (resp. $x_i = 0$) for all $i \in I$ and $x_i = 0$ (resp. 
$x_i = 1$) otherwise.

%%%%%%%%%%%%%%%%%%%%%%%%%%
\section{Results}
\label{sec:res}

\renewcommand{\arraystretch}{1.4}
\begin{table}[!t]
	\centerline{
		\begin{tabular}{|C{15mm}||C{12mm}|C{19mm}|C{20mm}|C{20mm}|C{20mm}|}
			\hline
			\backslashbox{ECA}{M} & \p & \bip & \bs & \bp & \lc\\
			\hline\hline
			156 &
        $\Theta(1)$ & 
        $\Theta(2^{\sqrt{n\log(n)}})$ & 
				$\Omega(2^{\sqrt{n\log(n)}})$ & 
				$\Omega(2^{\sqrt{n\log(n)}})$ & 
				$\Omega(2^{\sqrt{n\log(n)}})$\\\hline
      178 & 
				$\Theta(1)$ & 
				$\Theta(n)$ & 
				$O(n)$ & $\Omega(2^{\sqrt{n\log(n)}})$ & 
				$\Omega(2^{\sqrt{n\log(n)}})$\\\hline
		\end{tabular}
	}\bigskip
	\caption{Asymptotic complexity in terms of the length of the largest limit 
	  cycles of each of the two studied ECA rules, depending on the update 
	  modes.
  }
	\label{table:rules_and_updatemodes}
\end{table}

In this section, we will present the main results of our investigation related 
the asymptotic complexity of ECA rules 156 and 178, in terms of the lengths of 
their largest limit cycles depending on update modes considered. 
These results are summarized in Table~\ref{table:rules_and_updatemodes}
which highlights (a)synchronism sensitivity of ECA. 
As a reminder, these two rules have been chosen because they illustrate 
perfectly the very impact of the choice of update modes on their 
dynamics.\smallskip

By presenting the results rule by rule, we clearly compare the complexity 
changes brought up by the different update modes.
Starting with ECA rule $156$, we prove that it has limit cycles of length at 
most $2$ in parallel, but that in any other update mode can lead to reach 
limit cycles of superpolynomial length.
Then, we show that for ECA rule $178$, complexity increases less abruptly but 
still can reach very long cycles with carefully chosen update modes which fall 
into the category of update modes instantiating local function repetitions 
over a period.

\subsection{ECA rule $156$}

ECA rule $156$ is defined locally by a transition table which associates any 
local neighborhood configuration $(x_{i-1}^t, x_i^t, x_{i+1}^t)$ at step $t 
\in \N$ with a new state $x_i^{t+1}$, where $i \in \Z/n\Z$, as 
follows:\medskip
 
\renewcommand{\arraystretch}{1}
\centerline{$\begin{array}{|c||c|c|c|c|c|c|c|c|}
	\hline
	(x_{i-1}^t, x_i^t, x_{i+1}^t) & 000 & 001 & 010 & 011 & 100 & 101 & 110 & 
	  111\\
	\hline
	x_i^{t+1} & 0 & 0 & 1 & 1 & 1 & 0 & 0 & 1\\
	\hline
	\end{array}$}\medskip
	
Space-time diagrams depending on different update modes of a specific 
configuration under ECA rule $156$ are given in Figure~\ref{fig:rule156}. 
Each of them depicts a trajectory which gives insights about the role of walls 
together with the update modes in order to reach long limit cycles.

\begin{figure}[t!]
	\centerline{
  	% (156,PAR)	
	  \begin{minipage}{.17\textwidth}
    	\centerline{
      	\scalebox{.45}{
          \includegraphics[page=7]{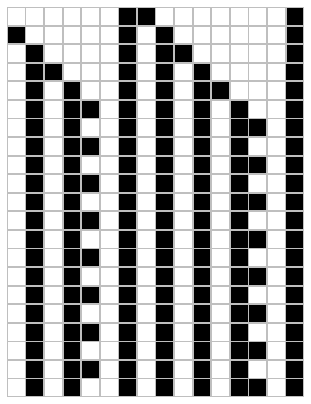}	
        }
      }
  	\end{minipage}
	  \quad
  	% (156,BIP)
	  \begin{minipage}{.17\textwidth}
    	\centerline{
    	  \scalebox{.45}{
          \includegraphics[page=8]{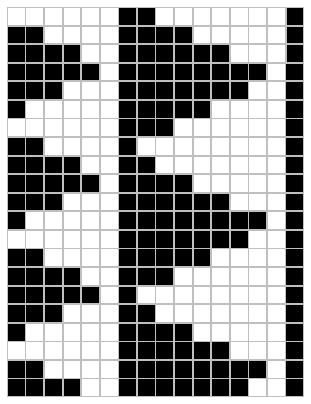}	
        }
      }
  	\end{minipage}
  	\quad
	  % (156,BS)
  	\begin{minipage}{.17\textwidth}
    	\centerline{
  	    \scalebox{.45}{
    	    \includegraphics[page=9]{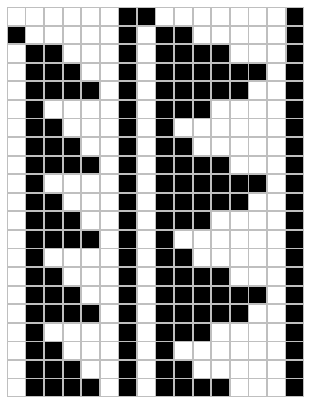}
        }
      }
  	\end{minipage}
	  \quad
  	% (156,BP)
	  \begin{minipage}{.17\textwidth}
    	\centerline{
	      \scalebox{.45}{
          \includegraphics[page=10]{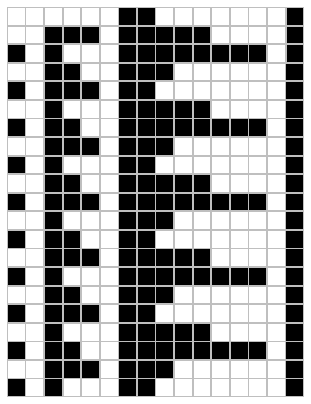}	
        }
      }
  	\end{minipage}
	  \quad
  	% (156,LC)
	  \begin{minipage}{.17\textwidth}
    	\centerline{
    	  \scalebox{.45}{
          \includegraphics[page=11]{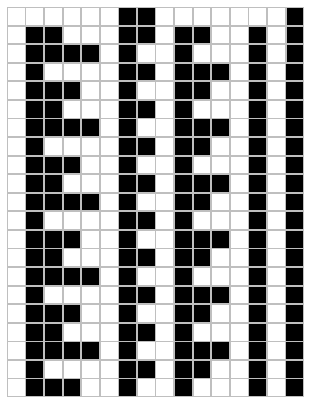}	
        }
      }
  	\end{minipage}
	}\smallskip
	
	\centerline{
	  \begin{minipage}{.17\textwidth}
    	\centerline{(a)}
  	\end{minipage}
	  \quad
  	\begin{minipage}{.17\textwidth}
    	\centerline{(b)}
  	\end{minipage}
	  \quad
  	\begin{minipage}{.17\textwidth}
	    \centerline{(c)}
  	\end{minipage}
  	\quad
	  \begin{minipage}{.17\textwidth}
    	\centerline{(d)}
  	\end{minipage}
	  \quad
  	\begin{minipage}{.17\textwidth}
	    \centerline{(e)}
  	\end{minipage}
	}
	\caption{Space-time diagrams (time going downward) of configuration 
		$0000001100000001$ following rule $156$ depending on: 
		(a) the parallel update mode $\mu_\p = (\integers{16})$, 
		(b) the bipartite update mode $\mu_\bip = (\{i \in \integers{16} \mid 
		  i \equiv 0 \mod 2\}, \{i \in \integers{16} \mid i \equiv 1 \mod 2\})$, 
		(c) the block-sequential update mode $\mu_\bs = 
			(\{10,15\},\{0,1,5,7,8,12\}, \{4,6,9,11,14\},\{3,13\},\{2\})$, 
		(d) the block-parallel update mode $\mu_\bp = 
		  \{(0,1),(2,3,4),(5),(6,8,7),(11,10,9),(14,13,12),(15)\}$, 
		(e) the local clocks update mode $\mu_\lc = 
  		(P = (2,2,2,2,4,4,4,4,3,3,3,3,1,4,1,1), 
	  	\Delta = (1,1,1,0,3,3,3,2,1,1,1,0,0,3,0,0))$.}
	\label{fig:rule156}
\end{figure}

\begin{lemma}
	\label{lem:156wall01}
	ECA rule $156$ admits only one wall, namely the word $w = 01$.
\end{lemma}

\begin{proof}
	For all $a, b \in \B$, $f_{156}(a01) = 0$ and $f_{156}(01b) = 1$. 
	By definition of the rule, no other word of length $2$ gets this property. 
	Thus $w$ is the unique wall for ECA rule $156$.\hfill\qed
\end{proof}

\begin{lemma}
	\label{lem:156wallcreation}
	ECA rule $156$ can dynamically create new walls only if the underlying 
	update mode makes two consecutive cells update their state simultaneously. 
\end{lemma}

\begin{proof}
	Remark that to create walls in the trajectory of a configuration $x 
	\in \B^n$, $x$ need to have at least one wall. 
	Indeed, the only configurations with no walls are $0^n$ and $1^n$, and they 
	are fixed points.
	Thus, with no loss of generality, let us focus on a subconfiguration $y \in 
	\B^k$, with $k \leq n-2$, composed of free state cells surrounded by two 
	walls.
	Notice that since the configurations are toric, the wall ``at the left'' and 
	``at the right'' of $y$ can be represented by the same two cells.\smallskip
	
	Configuration $y$ is necessarily of the form $y = (1)^\ell(0)^r$, with $k = 
	\ell + r$. 
	Furthermore, since $f_{156}(000) = 0$ and $f_{156}(100) = 1$, the only cells 
	whose states can change are those where $1$ meets $0$, i.e. $y_{\ell-1}$ and 
	$y_\ell$.
	Such state changes depends on the schedule of updates between these two 
	cells.
	Let us proceed with case disjunction:
	\begin{enumerate}
	\item Case of $\ell \geq 1$ and $r \geq 1$:
		\begin{itemize}
		\item if $y_\ell$ is updated strictly before $y_{\ell-1}$, then $y^1 = 
			(1)^{\ell+1}(0)^{r-1}$, and the number of $1$s (resp. $0$s) increases 
			(resp. decreases);
		\item if $y_\ell$ is updated strictly after $y_{\ell-1}$, then $y^1 = 
			(1)^{\ell-1}(0)^{r+1}$, and the number of $1$s (resp. $0$s) decreases 
			(resp. increases);
		\item if $y_\ell$ and $y_{\ell-1}$ are updated simultaneously, then $y^1 = 
			(1)^{\ell-1}(01)(0)^{r-1}$, and a wall is created.
		\end{itemize}
	\item Case of $\ell = 0$ (or of $y = 0^k$): nothing happens until $y_0$ is 
	  updated. 
		Let us admit that this first state change has been done for the sake of 
		clarity and focus on $y^1 = (1)^{\ell = 1}(0)^{r = k-1}$, which falls into 
		Case 1.
	\item Case of $r = 0$ (or of $y = 1^k$): symmetrically to Case 2, nothing 
  	happens until $y_{k-1}$ is updated. 
		Let us admit that this first state change has been done for the sake of 
		clarity and focus on $y^1 = (1)^{\ell = k-1}(0)^{1}$, which falls into 
		Case 1.
	\end{enumerate}
	As a consequence, updating two consecutive cells simultaneously is a 
	necessary condition for creating new walls in the dynamics of ECA rule 
	$156$.\qed 
\end{proof}

\begin{thm}
	\label{thm:156parallel}
	$(156, \p{})$ has only fixed points and limit cycles of length two.
\end{thm}

\begin{proof}
	We base the proof on Lemmas~\ref{lem:156wall01} 
	and~\ref{lem:156wallcreation}.
	So, let us analyze the possible behaviors between two walls, since by 
	definition, what happens between two walls is independent of what happens 
	between two other walls.\smallskip
	
	Let us prove the results by considering the three possible distinct cases 
	for a configuration between two walls $y \in \B^{k+4}$:
	\begin{itemize}
	\item Consider $y = (01)(0)^k(01)$ the configuration with only $0$s between 
	  the two walls. 
		Applying the rule twice, we obtain $y^2 = (01)^2(0)^{k-2}(01)$.
		So, a new wall appears every two iterations so that, for all $t < 
		\frac{k}{2}$, $y^{2t} = (01)^{t+1}(0)^{k-2t}(01)$, until a step is reached 
		where there is no room for more walls. 
		This step is reached after $k$ (resp $k-1$) iterations when $k$ is even 
		(resp. odd) and is such that there is only walls if $k$ is even (which 
		implies that the dynamics has converged to a fixed point), and only walls 
		except one cell otherwise.
		In this case, considering that $t = \lfloor \frac{k}{2} \rfloor$, because 
		$f_{156}(100) = 1$ and $f_{156}(110) = 0$, we have that 
		$y^{2t} = (01)^{t+1}(0)(01) \to y^{2t+1} = (01)^{t+1}(1)(01) \to 
		y^{2t} = (01)^{t+1}(0)(01)$, which leads to a limit cycle of length 
		$2$.
	\item Consider now that $y = (01)(1)^k(01)$. 
		Symmetrically, for all $t < \frac{k}{2}$, $y^{2t} = 
		(01)(0)^{k-2t}(01)^{t+1}$. 
		The same reasoning applies to conclude that the length of the largest 
		limit cycle is $2$. 
	\item Consider finally configuration $y = (01)(1)^\ell(0)^r(01)$ for which 
	  $k = \ell + r$.
		Applying $f_{156}$ on it leads to $y^1 = 
		(01)(1)^{\ell-1}(01)(0)^{r-1}(01)$, 
		which falls into the two previous cases.\hfill\qed
	\end{itemize}
\end{proof}

\begin{thm}
	\label{thm:156BIP}
	$(156, \bip{})$ of size $n$ has largest limit cycles of length 
	$\Theta(2^{\sqrt{n\log(n)}})$.
\end{thm}

\begin{proof}
	First, by definition of a bipartite update mode and by 
	Lemmas~\ref{lem:156wall01} and~\ref{lem:156wallcreation}, the only walls 
	appearing in the dynamics are the ones present in the initial configuration. 
	Let us prove that, given two walls $u$ and $v$, distanced by $k$ cells, the 
	largestlimit cycles of the dynamics between $u$ and $v$ are of length $k+1$. 
	We proceed by case disjunction depending on the nature of the 
	subconfiguration $y \in \B^{k+4}$, with the bipartite update mode 
	$\mu = (\{i \in \integers{k+4} \mid i \equiv 0 \mod 2\}, 
	\{i \in \integers{k+4} \mid i \equiv 0 \mod 2\})$:
	\begin{enumerate}
	\item Case of $y = (01)(0)^k(01)$: we prove that configuration 
  	$(01)(1)^k(01)$ is rea\-ched after $\lceil \frac{k}{2} + 1 \rceil$ 
  	steps:\medskip
  	
		\begin{minipage}{.4\textwidth}
			\begin{itemize}
			\item If $k$ is odd, we have:\\[-3mm]
				\begin{equation*}
					\begin{array}{ccc}
						y & = & (01)(0)^k(01)\\
						y^1 & = & (01)(1)(0)^{k-1}(01)\\
						y^2 & = & (01)111(0)^{k-3}(01)\\
						& \vdots & \\
						y^{\lfloor\frac{k}{2}\rfloor} & = & (01)(1)^{k-2}(00)(01)\\
						y^{\lceil\frac{k}{2}\rceil} & = & (01)(1)^{k}(01)\\
					\end{array}
				\end{equation*}
			\end{itemize}
		\end{minipage}
		\quad
		\begin{minipage}{.42\textwidth}
			\begin{itemize}
			\item If $k$ is even, we have:\\[-3mm]
				\begin{equation*}
					\begin{array}{ccc}
						y & = & (01)(0)^k(01)\\
						y^1 & = & (01)(1)(0)^{k-1}(01)\\
						y^2 & = & (01)(111)(0)^{k-2}(01)\\
						& \vdots & \\
			      y^{\frac{k}{2}} & = & (01)(1)^{k-1}(0)(01)\\
						y^{\frac{k}{2}+1} & = & (01)(1)^{k}(01)\\
					\end{array}
				\end{equation*}
			\end{itemize}			
		\end{minipage}\smallskip
		
	\item Case of $y = (01)(1)^k(01)$: we prove that configuration 
  	$(01)(0)^k(01)$ is rea\-ched after $\lceil \frac{k}{2} + 1 \rceil$ 
  	steps:\\
		\begin{minipage}{.4\textwidth}
			\begin{itemize}
			\item If $k$ is odd, we have:\\[-3mm]
				\begin{equation*}
					\begin{array}{ccc}
						y & = & (01)(1)^k(01)\\
						y^1 & = & (01)(1)^{k-1}(0)(01)\\
						y^2 & = & (01)(1)^{k-3}(000)(01)\\
						& \vdots & \\
						y^{\lfloor\frac{k}{2}\rfloor} & = & (01)(11)(0)^{k-2}(01)\\
						y^{\lceil\frac{k}{2}\rceil} & = & (01)(0)^{k}(01)\\
					\end{array}
				\end{equation*}
			\end{itemize}
		\end{minipage}
		\quad
		\begin{minipage}{.42\textwidth}
			\begin{itemize}
			\item If $k$ is even, we have:\\[-3mm]
				\begin{equation*}
					\begin{array}{ccc}
						y & = & (01)(1)^k(01)\\
						y^1 & = & (01)(1)^{k-2}(00)(01)\\
						y^2 & = & (01)(1)^{k-4}(0000)(01)\\
						& \vdots & \\
						y^{\frac{k}{2}} & = & (01)(11)(0)^{k-2}(01)\\
						y^{\frac{k}{2}+1} & = & (01)(0)^{k}(01)\\
					\end{array}
				\end{equation*}
			\end{itemize}
		\end{minipage}\smallskip
		
	\item Case of $y = (01)(1)^\ell(0)^r(01)$: this case is included in Cases 1 
	  and 2.
	 \end{enumerate}\smallskip

	As a consequence, the dynamics of any $y$ leads indeed to a limit cycle of 
	length $k+1$. 
	Remark that if we had chosen the other bipartite update mode as reference, 
	the dynamics of any $y$ would have been symmetric and led to the same limit 
	cycle.\smallskip
	
	Finally, since the dynamics between two pairs of distinct walls of 
	independent of each other, the asymptotic dynamics of a global configuration 
	$x$ such that $x = (01)(0)^{k_1}(01)(0)^{k_2}(01)\dots(01)(0)^{k_m}$ is a 
	limit cycle whose length equals to the least common multiple of the lengths 
	of all limit cycles of the subconfigurations $(01)(0)^{k_1}(01), 
	(01)(0)^{k_2}(01), \dots, (01)(0)^{k_m}(01)$.
	We derive that the largest limit cycles are obtained when the $(k_i+1)$s are 
	distinct primes whose sum is equal to $n - 2m$, with $m$ is constant.
	As a consequence, the length of the largest limit cycle is lower- and upper-
	bounded by the primorial of $n$ (\ie{} the maximal product of distinct 
	primes whose sum is $\leq n$.), denoted by function $h(n)$.
	In~\cite{Deleglise2015}, it is shown  in Theorem~$18$ that when $n$ tends to 
	infinity, $\log h(n) \sim \sqrt{n \log n}$.
	Hence,  we deduce that the length of the largest limit cycles of 
	$(156, \bip{})$ of size $n$ is $\Theta(2^{\sqrt{n\log(n)}})$.\hfill\qed
\end{proof}

\begin{cor}
	\label{cor:156other}
	The families $(156, \bs{})$, $(156, \bp{})$ and $(156, \lc{})$ of size 
	$n$ have largest limit cycles of length $\Omega(2^{\sqrt{n\log(n)}})$.
\end{cor}

\begin{proof}
	Since the bipartite update modes are specific block-sequential and 
	block-parallel update modes, and since both block-sequential and 
	block-parallel update modes are parts of local-clocks update modes, all of 
	them inherit the property stating that the lengths of the largest limit 
	cycles are lower-bounded by $2^{\sqrt{n\log(n)}}$.\hfill\qed
\end{proof}

\subsection{ECA rule $178$}

ECA rule $178$ is defined locally by the following transition table:\medskip

\centerline{$\begin{array}{|c||c|c|c|c|c|c|c|c|}
	\hline
	(x_{i-1}^t, x_i^t, x_{i+1}^t) & 000 & 001 & 010 & 011 & 100 & 101 & 110 & 
	  111\\
	\hline
	x_i^{t+1} & 0 & 1 & 0 & 0 & 1 & 1 & 0 & 1\\
	\hline
	\end{array}$}\medskip

Space-time diagrams depending on different update modes of a specific 
configuration under ECA rule $178$ are given in Figure~\ref{fig:rule178}. 
Each of them depict trajectories giving ideas of how to reach a limit cycle of 
high complexity.

\begin{figure}[t!]
	\centerline{
  	% (178,PAR)	
	  \begin{minipage}{.17\textwidth}
  	  \centerline{
	      \scalebox{.45}{
    	    \includegraphics[page=12]{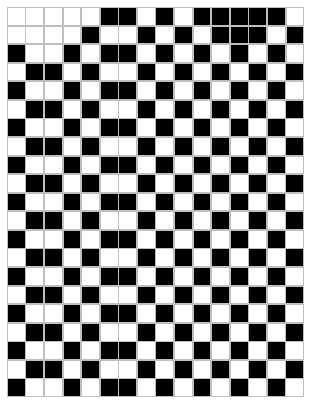}
    	  }
    	}
  	\end{minipage}
  	\quad
	  % (178,BIP)
  	\begin{minipage}{.17\textwidth}
	    \centerline{
      	\scalebox{.45}{
      	  \includegraphics[page=13]{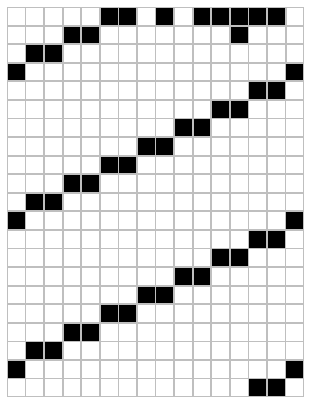}
      	}
     }
  	\end{minipage}
	  \quad
  	% (178,BS)
	  \begin{minipage}{.17\textwidth}
    	\centerline{
      	\scalebox{.45}{
        	\includegraphics[page=14]{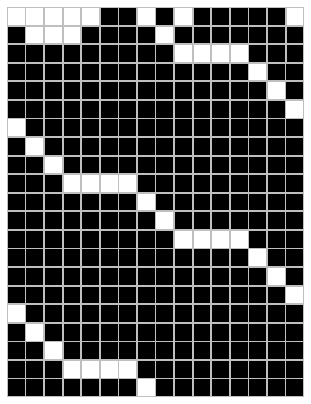}
        }
      }
  	\end{minipage}
	  \quad
  	% (178,BP)
	  \begin{minipage}{.17\textwidth}
    	\centerline{
	      \scalebox{.45}{
	        \includegraphics[page=15]{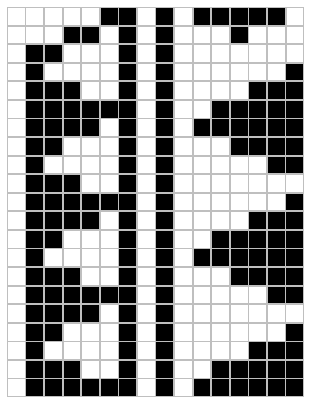}
        }
      }
  	\end{minipage}
	  \quad
  	% (178,LC)
	  \begin{minipage}{.17\textwidth}
    	\centerline{
      	\scalebox{.45}{
        	\includegraphics[page=16]{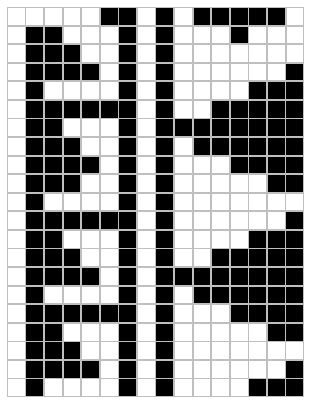}
        }
      }
  	\end{minipage}
	}\smallskip
	
	\centerline{
	  \begin{minipage}{.17\textwidth}
    	\centerline{(a)}
  	\end{minipage}
	  \quad
  	\begin{minipage}{.17\textwidth}
	    \centerline{(b)}
  	\end{minipage}
	  \quad
  	\begin{minipage}{.17\textwidth}
	    \centerline{(c)}
  	\end{minipage}
	  \quad
    \begin{minipage}{.17\textwidth}
    	\centerline{(d)}
	  \end{minipage}
  	\quad
  	\begin{minipage}{.17\textwidth}
	    \centerline{(e)}
  	\end{minipage}
	}
	\caption{Space-time diagrams (time going downward) of configuration 
		$0000011010111110$ following rule $178$ depending on: 
		(a) the parallel update mode $\mu_\p = (\integers{16})$, 
		(b) the bipartite update mode $\mu_\bip = (\{i \in \integers{16} \mid i 
		  \equiv 0 \mod 2\}, \{i \in \integers{16} \mid i \equiv 1 \mod 2\})$, 
		(c) the block-sequential update mode $\mu_\bs = (\{3,9,15\}, 
		  \{2,4,8,10,14\}, \{11,5,7,11,13\},\{0,6,12\})$, 
		(d) the block-parallel update mode $\mu_\bp =\{(0), (1), (2,3), (4,5), 
		  (6), (7), (8), (9), (10, 11), (12, 13), (14, 15)\}$, 
		(e) the local clocks update mode $\mu_\lc = 
			(P = (1,1,2,2,2,2,1,1,1,4,4,4,4,4,4,4), 
			\Delta = (0,0,1,0,1,0,0,0,0,0,1,0,1,0,1,0))$.}
	\label{fig:rule178}
\end{figure}

\begin{lemma}
	\label{lem:178walls}
	ECA rule $178$ admits two walls, $01$ and $10$, which are relative walls.
\end{lemma}

\begin{proof}
	Let $a, b \in \B$. %such that $a \neq b$. 
	Notice that by definition of the rule: 
	$f(a00) \neq f(b00)$ and $f(00a) \neq f(00b)$, and 
	$f(a11) \neq f(b11)$ and  $f(11a) \neq f(11b)$ if $a \neq b$; 
	$f(a01) = 1$ and $f(01a) = 0$; and
	$f(a10) = 0$ and $((10a) = 1$.
	Thus, neither $00$ nor $01$ nor $10$ nor $11$ are absolute walls. 
	From what precedes, notice that the properties of $00$ and $11$ prevent 
	them to be relative walls. 
	Consider now the two words $u = 01$ and $v = 10$ and let us show that they 
	constitute relative walls. 
	Regardless the states of the cells surrounding $u$, every time both 
	cells of $u$ are updated simultaneously, $u$ changes to $v$ and similarly, 
	$v$ will change to $u$ independently of the states of the cells that 
	surround it, as long as both its cells are updated together.
	Thus, $u$ and $v$ are relative walls.\hfill\qed
\end{proof}

We will show certain $\bp$ and $\lc$ update modes that are able to produce 
these relative walls.

\begin{thm}
	\label{thm:178BIP}
	Each representative of $(178, \bip{})$ of size $n$ has largest limit cycles 
	of length $\Theta(n)$.
\end{thm}

\begin{proof}
	By Lemma~\ref{lem:178walls} (and its proof), $(178, \bip{})$ has no walls 
	since there is no way of updating two consecutive cells simultaneously. 
	Since configurations $0^n$ and $1^n$ are fixed points, let us consider other 
	configurations and denote by $\mu_\bip{}$ (resp. $\barmu_\bip{}$) the 
	update mode defined by $(\{i \in \integers{n} \mid i \equiv 0 \mod 2\}, 
	\{i \in \integers{n} \mid i \equiv 1 \mod 2\})$ (resp. the other one). 
	Let us begin with configurations $x$ composed by one isle $1$s.
	With no loss of generality, let us take $x = (1)^\ell(0)^r$, with $n = 
	\ell + r \equiv 0 \mod 2$ by definition of a bipartite update mode:
	\begin{itemize}
	\item If $\ell$ is odd, with $\mu_\bip{}$, we have:
		\begin{equation*}
			\begin{array}{ccl}
				x & = & (1)^\ell(0)^r\\
				x^1 & = & (0)^2(1)^{\ell - 4}(0)^{r+2}\\
				& \vdots & \\
				x^i & = & (0)^{2i}(1)^{\ell - 4i}(0)^{r + 2i}\\
				& \vdots & \\				
				x^{\lfloor \frac{\ell}{4} \rfloor} & = & \begin{cases}
						(0)^{\lfloor \frac{\ell}{2} \rfloor} (1) 
							(0)^{\lfloor \frac{\ell}{2} \rfloor} & \text{if } \ell 
							\text{ is the smallest of the two odd numbers}\\
							& \text{with equal } \lfloor \frac{\ell}{4} \rfloor\\
						(0)^{\lfloor \frac{\ell}{2} \rfloor} (1)^3 
							(0)^{\lfloor \frac{\ell}{2} \rfloor} & \text{otherwise}\\
					\end{cases}\\
				x^{\lfloor \frac{\ell}{4} \rfloor+1} & = & (0)^n\text{.}
			\end{array}
		\end{equation*}
		With $\barmu_\bip{}$, with a similar reasoning, we get the symmetric 
		result with the number of $1$s increasing by $4$ at each step.
		Consequently, such configurations lead to fixed points, either $0^n$ or 
		$1^n$.
	\item If $\ell$ is even, let us show that the isle of $1$s shifts over 
	  time to the right with $\mu_\bip{}$ and to the left with $\barmu_\bip{}$.
		With $\mu_\bip{}$, we have:
		\begin{equation*}
			\begin{array}{ccl}
				x & = & (1)^\ell(0)^r\\
				x^1 & = & (0)(2)^\ell(0)^{r-2}\\
				& \vdots & \\
				x^i & = & (0)^{2i}(1)^\ell(0)^{r-2i}\\
				& \vdots & \\
			\end{array}
		\end{equation*}
		\begin{equation*}
			\begin{array}{ccl}
				x^{\frac{n}{2}-1} & = & (1)^{\ell-2}(0)^r(1)^2\\
				x^{\frac{n}{2}} & = & (1)^\ell(0)^r\text{.}
			\end{array}
		\end{equation*}
		With $\barmu_\bip{}$, with a similar reasoning, we get the symmetric 
		result, the isle of $1$s shifting to the left.
		Consequently, such configurations lead to limit cycles of length 
		$\frac{n}{2}$.
	\end{itemize}

	Now, let us consider configurations with several isles of $1$s.
	With no loss of generality, let us focus on configurations with two isles 
	of $1$s because they capture all the possible behaviors. 
	There are $10$ distinct cases which depends on both the parity of the size 
	of the isles of $1$s and the parity of the position of the first cell of 
	each isle. 
	This second criterion coincides locally with the nature of the bipartite 
	update mode. 
	Indeed, given an isle of $1$s, if its first cell $i$ is even (resp. odd), 
	then the isle follows $\mu_\bip{}$ (resp. $\barmu_\bip{}$) locally. 
	For the sake of clarity, let us make us of the following notation: 
	we use $\alpha \in \{\even, \odd\}$ to denote the parity of the size of the 
	isles of $1$s, and $\beta \in \{\mu, \barmu\}$ to denote the local 
	bipartite update mode followed by the isles. 
	Let us proceed by case disjunction, where the cases are denoted by a pair of 
	criterion of each sort, one for the first isle of $1$s, another one for the 
	second.	
	\begin{itemize}
	\item If the two isles of $1$s are of even sizes:
		\begin{enumerate}
		\item Case $(e_\mu, e_\mu)$:
			By what precedes, because the two isles are separated by each other by 
			an even number of cells at least equal to $2$, they both shift to the 
			right with an index equal to $2$ over time. 
			Thus, such a configuration lead to a limit cycle of length 
			$\frac{n}{2}$.
		\item Case $(e_{\barmu}, e_{\barmu})$:
			This case is similar to the previous one, except that the isles shift 
			to the right.
		\item Case $(e_\mu, e_{\barmu})$:
			By what precedes, the first isle shifts to right as well as the second 
			isle shifts to the left over time, both with an index $2$.
			They do so synchronously at each time step until the isles meet. 
			Notice that the two isles are followed by an odd number of $0$s.
			So, let us consider that initial configuration 
			$x = (1)^k(0)^\ell(1)^{k'}(0)^r$, with $n = k + k' + \ell + r$.
			The two isles inevitably meet from a configuration 
			$x^t = (1)^k(0)^{\ell'}(1)^{k'}(0)^{r'}$ 
			(up to rotation of the configuration on the torus), where $k, k', \ell', 
			r', t \in \Z$ and $n = k + k' + \ell' + r'$ such that $k$ and $k'$ are 
			even and represent the size of each isle, $\ell'$ and $r'$ are odd, and 
			$\ell' = 3$ (resp. $\ell' = 1$) if the initial number of $0$s following 
			the first isle in $x^0$ is the smallest (resp. the biggest) of the odd 
			numbers with equal $\lfloor \frac{\ell}{4} \rfloor$. 
			Then, we have:
			\begin{equation*}
				\begin{array}{ccc}
					x^t & = & (1)^k(0)^{\ell'}(1)^{k'}(0)^{r'}\\
					x^{t+1} & = & (0)^2(1)^{k+k'-1}(0)^{r'+2}\text{.}
				\end{array}
			\end{equation*}
			Configuration $x^{t+1}$ is then composed of a unique isle of $1$s of 
			odd size whose first cell position is even, \ie{} this isle follows 
			$\mu_\bip{}$ locally and converges to $0^n$ as proven above.
			Remark that the case $(e_{\barmu}, e_\mu)$ is strictly equivalent 
			because of the toric nature of the ECA.
		\end{enumerate}
	\item If the two isles of $1$s are of odd sizes:
		\begin{enumerate}
		\item Case $(o_\mu, o_\mu)$:
			By what precedes, the two isles evolve locally by decreasing their 
			numbers of $1$s until $x$ becomes the fixed point $0^n$.
		\item Case $(o_\mu, o_{\barmu})$:
			Notice first that by definition, the two isles are inevitably separated 
			by even numbers $\ell \geq 2$ and $r \geq 2$ of $0$s on each side.
			Let us consider configuration $x = (1)^k(0)^\ell(1)^{k'}(0)^r$ such that 
			$n = k + k' + \ell + r$.
			By what precedes, the first isle evolves locally towards 
			$O^k$ by losing four $1$s (two from each side) at each step. 
			Conversely, the second isle increases its number of 
			$1$s by four (two more on each side) at each step. 
			Consequently, these two isles of $1$s never meet and the first isle 
			spreads it $1s$ over time until $x$ reaches $1^n$.
			Remark that the case $(o_{\barmu}, o_\mu)$ is strictly equivalent 
			because of the toric nature of the ECA.
		\item Case $(o_{\barmu},o_{\barmu})$:
			By what precedes, the two isles increases by four their number of 
			$1$s (two on each side).
			Inevitably, there exists a step $t$ at which they meet and become a 
			unique isle of $1$s such that the position of its first cell position 
			is odd, because of the parity of the increasing of $1$s, and follows 
			thus $\barmu_\bip{}$. 
			As a consequence, this isle spreads its $1$s and $x$ evolves over time 
			until it reaches $1^n$.
		\end{enumerate}
	\item If the sizes of the two isles of $1$s are of distinct parity:
		\begin{enumerate}
		\item Case $(e_\mu,o_\mu)$:
			First, notice that configuration $x$ is necessary such that there are an 
			even (resp. odd) number of $0$s which follow the first (resp. second) 
			isle. 
			So, let us consider $x = (1)^k(0)^\ell(1)^{k'}(0)^r$, with $\ell \geq 2$ 
			even and $r$ odd, such that $n = k + k' + \ell + r$. 
			By what precedes, the first isle shifts to the right with index $2$ at 
			each time step.
			The second isle reduces its its number of $1$s by four at each step 
			($two$ on each side).
			Thus, the two isles never meet: the second isle converges locally to 
			$0^{k'}$ while the first isle keeps shifting with an index $2$ over 
			time.
			Thus, $x$ reaches a limit cycle of length $\frac{n}{2}$.
		\item Case $(e_\mu,o_{\barmu})$:
			By what precedes, the first isle dynamically shifts to the right with 
			an index $2$ and the second increases its number of $1$s by four (two on 
			each side) at each step.
			Notice that configuration $x$ is necessary such that there are an odd 
			(resp. even) number of $0$s which follow the first (resp. second) isle. 
			So, let us consider that initial configuration 
			$x = (1)^k(0)^\ell(1)^{k'}(0)^r$, with $n = k + k' + \ell + r$, and 
			$\ell$ odd and $r \geq 2$ even.
			The two isles inevitably meet from a configuration 
			$x^t = (1)^k(0)^{\ell'}(1)^{k''}(0)^{r}$ (up to rotation of the 
			configuration on the torus), where $k, k'', \ell', r, t \in \Z$ and $n = 
			k + k'' + \ell' + r$ such that $k$ and $k''$ are even and represent the 
			size of each isle at step $t$, and $\ell' = 3$ (resp. $\ell' = 1$) if 
			the initial number of $0$s following the first isle in $x^0$ is the 
			smallest (resp. the biggest) of the odd numbers with equal 
			$\lfloor \frac{\ell}{4} \rfloor$. 
			Then, we have:
			\begin{equation*}
				\begin{array}{ccc}
					x^t & = & (1)^k(0)^{\ell'}(1)^{k''}(0)^{r}\\
					x^{t+1} & = & (0)^2(1)^{n-r}(0)^{r-2}\text{.}
				\end{array}
			\end{equation*}
			Configuration $x^{t+1}$ is then composed of a unique isle of $1$s of 
			even size whose first cell position is even, \ie{} this isle follows 
			$\mu_\bip{}$ locally, and shifts to the right with index $2$. 
			Thus, $x$ reaches a limit cycle of length $\frac{n}{2}$.
		\item Case $(e_{\barmu},o_\mu)$:
			Applying on this case the same reasoning as in the case 
			$(e_\mu,o_\mu)$ allows to show a dynamics which is symmetric, in the 
			sense that such a configuration evolves asymptotically towards a limit 
			cycle characterized by the first isle of $1$s shifting to the left with 
			index $2$, which confirms the reachability of a limit cycle of length 
			$\frac{n}{2}$.
		\item Case $(e_{\barmu},o_{\barmu})$:
			Applying on this case the same reasoning as in the case 
			$(e_\mu,o_{\barmu})$ allows to show a dynamics which symmetric, in the 
			sense that such a configuration evolves asymptotically towards a limit 
			cycle characterized by a unique isle of $1$s of size $n-\ell$ shifting 
			to the left with index $2$, which confirms the reachability of a limit 
			cycle of length $\frac{n}{2}$.		
		\end{enumerate}
	\end{itemize}
	All these cases taken together show that the largest limit cycle reachable 
	by $(178, \bip{})$ is of length $\frac{k}{2} = O(n)$. \qed
\end{proof}

\begin{thm}\label{thm:178bp}
  The family $(178,\bp)$ of size $n$ has largest limit cycles of length 
  $\Omega\left(2^{\sqrt{n\log n}}\right)$
\end{thm}

\begin{proof}
  Notice that the configurations of interest here are those having at least 
  one relative wall since $(0)^n$ and $(1)^n$ are fixed points. 
  By Lemma~\ref{lem:178walls} (and its proof) $w_1$ and $w_2$ are relative 
  walls for the family $(178,\bp)$ since there exist block-parallel updates 
  modes guaranteeing that the two contiguous cells carrying $w_1$ (resp. 
  $w_2$) are updated simultaneously and an even number of substeps over the 
  period. 
  Thus, let us consider in this proof an initial configuration $x$ with at 
  least one wall. 
  The idea is to focus on what can happen between two walls because the 
  dynamics of two subconfigurations delimited by two distinct pairs of walls 
  are independent from each other.\smallskip
  
  Now, let $y$ (resp. $y'$) be a subconfiguration of size $k+4$ such that 
  $y = (w_{\ell},y_2,\dots,y_{k+1},w_r)$, with 
  $w_{\ell}=y_0y_1,w_r=y_{k+2}y_{k+3}$ two relative walls in $W = \{w_1,w_2\}$ 
  and such that for all $i\in\{2,\dots,k+1\},y_i=0$ (resp. $y_i'=1$) and the 
  block-parallel mode
  \begin{equation*}
    \begin{split}
      \mu_{\bp} &= \{(y_0),(y_1),(y_2,y_3),\dots,(y_k,y_{k+1}),(y_{k+2}),
        (y_{k+3})\}\\
      &\equiv (\{y_0,y_1,y_2,y_4,\dots,y_{k},y_{k+2},y_{k+3}\},
        \{y_0,y_1,y_3,y_5,\dots,y_{k+1},y_{k+2},y_{k+3}\}),
    \end{split}
  \end{equation*}
  if $k$ is even, and:
  \begin{equation*}
    \begin{split}
      \mu_{\bp} &= \{(y_0),(y_1),(y_2,y_3),\dots,(y_k,y_{k+1}),(y_{k+2}),
        (y_{k+3})\}\\
      &\equiv (\{y_0,y_1,y_2,y_4,\dots,y_{k+1},y_{k+2},y_{k+3}\},
        \{y_0,y_1,y_3,y_5,\dots,y_{k},y_{k+2},y_{k+3}\}),
    \end{split}
  \end{equation*}    
  if $k$ is odd.\medskip
  
  The dynamics of $y$ follows four cases:
  \begin{enumerate}
  \item $w_{\ell}=w_r=01$, given $\tau\in\Z$, denoting the subconfigurations 
    obtained at a substep by $y^{\frac{\tau}{2}}$:\medskip
    
    \overfullrule=0pt
    \begin{minipage}{.48\textwidth}
      -- If $k$ is even, we have
        \begin{align*}
          y&=&(01)(0)^k(01)\\
          y^{\left(\frac{1}{2}\right)}&=&(10)(1)(0)^{k-1}(10)\\
          y^1&=&(01)(1)^2(0)^{k-3}(1)(01)\\
          y^{\left(\frac{3}{2}\right)}&=&(10)(1)^3(0)^{k-5}(1)^2(10)\\
          y^2&=&(01)(1)^4(0)^{k-7}(1)^3(01)\\
          &\vdots&\\    
          y^i&=&(01)(1)^{2i}(0)^{k-4i+1}(1)^{2i-1}(01)\\
          &\vdots&\\
          y^{\frac{k-1}{2}}&=&(01)(1)^k(01)\\
          y^{\frac{k}{2}-1}&=&(10)(1)^k(10)\\
          y^{\frac{k}{2}}&=&(01)(1)^k(01)\\                
        \end{align*}
    \end{minipage}        
    \begin{minipage}{.48\textwidth}
      -- If $k$ is odd, we have
        \begin{align*}
          y&=&(01)(0)^k(01)\\
          y^{\left(\frac{1}{2}\right)}&=&(10)(1)(0)^{k-1}(10)\\
          y^1&=&(01)(1)^2(0)^{k-2}(01)\\
          &\vdots&\\
          y^i&=&(01)(1)^{2i}(0)^{k-2i}(01)\\
          &\vdots&\\    
          y^{\frac{k+1}{2}}&=&(01)(1)^k(01)\\
          y^{\frac{k+2}{2}}&=&(10)(1)^{k-1}(0)(10)\\
          y^{\frac{k+3}{2}}&=&(01)(1)^{k-2}(0)^2(01)\\
          &\vdots&\\
          y^{k+1}&=&(01)(0)^k(01)\\                
        \end{align*}
    \end{minipage}
    Thus, subconfiguration $y$ converges towards fixed point $(01)(1)^k(01)$ 
    when $k$ is even and leads to a limit cycle of length $k+1$ when $k$ is 
    odd.
  \item $w_{\ell}=w_r=10$: taking $y'$ as the initial subconfiguration, and 
    applying the same reasoning, we can show that this case is analogous to 
    the previous one up to a symmetry, which allows us to conclude that $y'$ 
    converges towards fixed point $(10)(0)^k(10)$ when $k$ is even and leads 
    to a limit cycle of length $k+1$ when $k$ is odd.
  \item $w_{\ell}=01$ and $w_r=10$:\\[2mm]
    \begin{minipage}{.48\textwidth}
      -- If $k$ is even, we have
        \begin{align*}
          y&=&(01)(0)^k(10)\\
          y^{\left(\frac{1}{2}\right)}&=&(10)(1)(0)^{k-1}(01)\\
          y^1&=&(01)(1)^2(0)^{k-2}(10)\\
          &\vdots&\\
          y^i&=&(01)(1)^{2i}(0)^{k-2i}(10)\\
          &\vdots&\\
          y^{\frac{k}{2}}&=&(01)(1)^k(10)\\
          y^{\left(\frac{k+1}{2}\right)}&=&(10)(1)^k(01)&\\
          y^{\left(\frac{k+2}{2}\right)}&=&(01)(1)^{k-1}(0)(10)\\
          &\vdots&\\
          y^{k+1}&=&(01)(0)^k(10)\\            
        \end{align*}
    \end{minipage}
    \begin{minipage}{.48\textwidth}
      -- If $k$ is odd, we have
        \begin{align*}
          y&=&(01)(0)^k(10)\\
          y^{\left(\frac{1}{2}\right)}&=&(10)(1)(0)^{k-2}(01)\\
          y^1&=&(01)(1)^2(0)^{k-4}(1)^2(10)\\
          y^{\left(\frac{3}{2}\right)}&=&(10)(1)^3(0)^{k-6}(1)^3(01)\\
          y^2&=&(01)(1)^{4}(0)^{k-8}(1)^4(01)\\
          &\vdots&\\
          y^i&=&(01)(1)^{2i}(0)^{k-4i}(1)^{2i}(10)\\
          &\vdots&\\
          y^{\left(\frac{k-1}{2}\right)}&=&(01)(1)^k(10)\\
          y^{\frac{k}{2}}&=&(10)(1)^k(01)\\
          y^{\frac{k+1}{2}}&=&(01)(1)^k(10)\\            
        \end{align*}
    \end{minipage}
    Thus, subconfiguration $y$ leads to a limit cycle of length $k+1$ when $k$ 
    is even and converges towards fixed point $(01)(1)^k(10)$ when $k$ is odd.
  \item $w_{\ell}=10$ and $w_r=10$: taking $y'$ as the initial 
    subconfiguration, and applying the same reasoning we can show that this 
    case is analogous to the previous one up to a symmetry, which allows us to 
    conclude that $y'$ leads to a limit cycle of length $k+1$ when $k$ is even 
    and converges towards fixed point $(10)(0)^k(01)$ when $k$ is odd.
  \end{enumerate} 
  
  Since the dynamics between two pairs of distinct walls is independent of 
  each other, the asymptotic dynamics of a global configuration $x$ is a limit 
  cycle whose length equals the least common multiple of the lengths of all 
  limit cycles of the subconfigurations embedded into pairs of walls. 
  With the same argument as the one used in the proof of 
  Theorem~\ref{thm:156BIP}, we derive that the length of the largest limit 
  cycles of the family $(178,\bp)$ applied over a grid of size $n$ is 
  $\Omega\left(2^{\sqrt{n\log n}}\right)$.\hfill\qed
\end{proof}

\begin{thm}\label{thm:178SEQ}
    $(178,\seq)$ of size $n$ has largest limit cycles of length $O(n)$.
\end{thm}

\begin{proof}
  Let us fix $n \in \N$ and let us call $f$ the $178$-ECA rule. 
  Let us also fix a sequential update mode $\mu$. 
  We recall the notation $\mu(n) = \mu_n \in \integers{n}$ to denote the node 
  that will be updated at time $t=n$. 
  We will refer to the global rule under the update mode $\mu$ as 
  $F$.\smallskip 

  First, observe that, for the rule $178$, since $f(000) = 0$ and $f(111) = 1$ 
  then, the configurations $1^*$ and $0^*$ are fixed points for the parallel 
  update mode. 
  Thus, we have that they are also fixed points for the sequential update 
  modes, i.e $f_{\mu_{n}}(111) = 1$ and $f_{\mu_{n}}(000) = 0$ for all $n \in 
  \N$ and $F_{178}(1^*) = 1^*, F_{178}(0^*) = 0^*$.\smallskip

  Let us consider an interval $I_0 = [p,q] \subseteq \integers{n}$ and a 
  configuration $x$ such that $x_i = 1$ for all $i \in I_0$ and $x_i = 0$ 
  otherwise.
  We distinguish four different cases depending on the update mode:
  \begin{enumerate}
  \item Cell $p-1$ updates before $p$ (formally $\mu_{p-1}<\mu_p$), and $q+1$ 
    before $q$ ($\mu_q>\mu_{q+1}$).

    In that case, we have a situation as the one illustrated in the following 
    table:\smallskip
    
    \begin{center}
      \begin{tabular}{c|c|c|c|c|c|c|c|c|c|c}
        \hline
        $\dots$ & $p-2$ & $p-1$ & $p$ & $p+1$ & $\dots$ & $q-1$ & $q$ & $q+1$ 
        & $q+2$ & $\dots$\\
        \hline\hline
        $\dots$ & $0$ & $0$ & $1$ & $1$ & $\dots$ & $1$ & $1$ & $0$ & $0$ & 
        $\dots$\\
        \hline
        $\dots$ & $0$ & $1$ & $1$ & $1$ & $\dots$ & $1$ & $1$ & $1$ & $0$ & 
        $\dots$\\
        \hline
      \end{tabular}
    \end{center}\medskip
    
    We will prove by induction that in this case the state of every cell of 
    the configuration is $1$ after a finite number of steps. 
    Since the patterns $001$ and $100$ are such that $f(001) = 1$ and $f(100) 
    = 1$ then, we have that the nodes $p-1$ and $q+1$ will change their states 
    to $1$. 
    Let us define $\ell_1 = \max\{i < p : \mu_{i-1} < \mu_i \ \wedge \ \mu_i > 
    \mu_{i+1}\}$ and $r_1 = \min\{i > q : \mu_{i-1} < \mu_i \ \wedge \ \mu_i > 
    \mu_{i+1}\}$, and consider $L_1 = [\ell_1, p-1]$ and $R_1 = [q+1,r_1]$.

    \begin{center}
      \begin{tabular}{c|c|c|c|c|c|c|c|c|c|c}
        \hline
        $\ell-1$ & $\ell$ & $\dots$ & $p-1$ & $p$ & $\dots$ & $q$ & $q+1$ & 
        $\dots$ & $r$ & $r+1$\\
        \hline\hline
        $0$ & $0$ & $\dots$ & $0$ & $1$ & $\dots$ & $1$ & $0$ & $\dots$ & $0$ 
        & $0$\\
        \hline
        $0$ & $1$ & $\dots$ & $1$ & $1$ & $\dots$ & $1$ & $1$ & $\dots$ & $1$ 
        & $0$\\
        \hline
      \end{tabular}
    \end{center}

    We have that $F(x) = y$ where for $y$, there exists an interval $I_1 = L_1 
    \cup I_0 \cup R_1$ such that $I_0 \subseteq I_1$ and that $y_i = 1$ for 
    all $i \in I_1$ and $y_i = 0$ otherwise. 
    In addition, because of the assumptions of this case, we have that the 
    inclusion is strict, i.e. $|I_1| > |I_0|$, and that $\mu_{\ell_1-1} < 
    \mu_{\ell_1}$ and $\mu_r > \mu_{r+1}$.\smallskip

    Now, we are going to show that for each $k \geq 1$, there exists a 
    sequence of intervals $I_1, \hdots, I_k$ such that $I_1 \subseteq I_2 
    \subseteq \hdots \subseteq I_k$, $|I_1| < \hdots < |I_k|$, and such that 
    for each $k \in \integers{n}$, we have that:\smallskip

    \begin{enumerate}
      \item $I_{k} = [\ell_k,r_k]$ 
      \item $I_0 \subsetneq I_1 \subsetneq \dots \subsetneq I_k$.
    \end{enumerate}\smallskip
    
    We proceed by induction. 
    For $k = 1$, the base case is given by the latter construction. 
    Now assume that there exists a sequence of $I_k$ such that they have 
    properties (a) and (b).
    Then, let us define $\ell_{k+1} = \max\{i < \ell_k : \mu_{i-1} < \mu_{i} \ 
    \wedge \ \mu_{i} > \mu_{i+1}\}$ and $r_{k+1} = \min \{i > r_k : \mu_{i-1} 
    < \mu_{i} \ \wedge \ \mu_{i} > \mu_{i+1}\}$. 
    Both $\ell_{k+1}$ and $r_{k+1}$ are well defined because if 
    $\mu_{\ell_k-1} < \mu_{\ell_k}$, then $\ell_{k+1} = \ell_k-1$ and, 
    similarly, if $\mu_{r_k} > \mu_{r_k+1}$, then $r_{k+1} = r_k+1$. 
    From there, we can define $L_{k+1} = [\ell_{k+1} , \ell_k-1]$ and 
    $R_{k+1} = [r_k+1,r_{k+1}]$. 
    Indeed, if $I_{k+1} = R_{k+1} \cup I_k \cup L_{k+1}$, then 
    $I_{k+1} = [\ell_{k+1}, r_{k+1}]$, and $I_k \subsetneq I_{k+1}$.
    And because $|I_k| < |I_{k+1}|$ after $t$ iterations of $F_{178}$ ($t<n$), 
    the configuration reaches an homogeneous fixed point in which the state of 
    each cell will be equal to $1$, i.e. $F_{178}^t(x)=1^*$.

  \item $\mu_{p-1} > \mu_p$ and $\mu_q < \mu_{q+1}$. 
    This case is similar to Case~1, considering an isle of $0$s surrounded by 
    $1$s.
    \begin{center}
      \begin{tabular}{c|c|c|c|c|c|c|c|c|c|c}
        \hline
        $\dots$ & $p-2$ & $p-1$ & $p$ & $p+1$ & $\dots$ & $q-1$ & $q$ & $q+1$ 
        & $q+2$ & $\dots$\\
        \hline\hline
        $\dots$ & $0$ & $0$ & $1$ & $1$ & $\dots$ & $1$ & $1$ & $0$ & $0$ & 
        $\dots$\\
        \hline
        $\dots$ & $0$ & $0$ & $0$ & $1$ & $\dots$ & $1$ & $0$ & $0$ & $0$ & 
        $\dots$\\
        \hline
      \end{tabular}
    \end{center}
    This means that Case~2 leads always to the homogeneous fixed point $0^n$.
    
  \item $\mu_{p-1}>\mu_p$ and $\mu_q>\mu_{q+1}$.
    \begin{center}
      \begin{tabular}{c|c|c|c|c|c|c|c|c|c|c}
        \hline
        $\dots$ & $p-2$ & $p-1$ & $p$ & $p+1$ & $\dots$ & $q-1$ & $q$ & $q+1$ 
        & $q+2$ & $\dots$\\
        \hline\hline
        $\dots$ & $0$ & $0$ & $1$ & $1$ & $\dots$ & $1$ & $1$ & $0$ & $0$ & 
        $\dots$\\
        \hline
        $\dots$ & $0$ & $0$ & $0$ & $1$ & $\dots$ & $1$ & $1$ & $1$ & $0$ & 
        $\dots$\\
        \hline
      \end{tabular}
    \end{center}
    According to the analysis made for Case~1, we know that the configuration 
    gains $1$s to the right of the isle until it reaches a cell that we call 
    $r$. 
    We need to find how the dynamics will behave to the left.\smallskip
    
    If $p+1$ updates after $p$ ($\mu_p<\mu_{p+1}$), then cell $p+1$ will 
    become $0$. 
    Thus, let $r_1' \in \mathbb{N}$ be such that $r_1' = \min \{s > p : 
    \mu_{i-1} < \mu_{i} \ \wedge \ \mu_{i} > \mu_{i+1}\}$.
    Then, we will lose $1$s until we reach cell $r'$.
    \begin{center}
      \begin{tabular}{c|c|c|c|c|c|c|c|c|c|c|c}
        \hline
        $p-1$ & $p$ & $p+1$ & $\dots$ & $r'$ & $r'+1$ & $\dots$ & $q$ & $q+1$ 
        & $\dots$ & $r$ & $r+1$\\
        \hline\hline
        $0$ & $1$ & $1$ & $\dots$ & $1$ & $1$ & $\dots$ & $1$ & $0$ & $\dots$ 
        & $0$ & $0$\\    
        \hline
        $0$ & $0$ & $0$ & $\dots$ & $0$ & $1$ & $\dots$ & $1$ & $1$ & $\dots$ 
        & $1$ & $0$\\
        \hline
      \end{tabular}
    \end{center}
    Note that both $r_1$ and $r_1'$ are defined by the fact that they update 
    after the cell to their right ($\mu_{r_1} > \mu_{r_1+1}$ and $\mu_{r_1'} > 
    \mu_{r_1'+1}$). 
    Similarly, we can define $p_1 = r_1'+1$ and $q_1 = r_1$, and in turn, find 
    new $r'_2$ and $r_2$ defined by $r_2' = \min \{s > r_1' : \mu_{s-1} < 
    \mu_{s} \ \wedge \ \mu_{s} > \mu_{s+1}\}$ and $r_1 = \min \{s > r_1 : 
    \mu_{s-1} < \mu_{s}\ \wedge \ \mu_{s} > \mu_{s+1}\}$.
    Recursively, we define $r_k' = \min \{s > r_{k-1}' : \mu_{s-1} < \mu_{s} \ 
    \wedge \ \mu_{s} > \mu_{s+1}\}$ and $r_k = \min \{s > r_{k-1} : \mu_{s-1} 
    < \mu_{s} \ \wedge \ \mu_{s} > \mu_{s+1}\}$ and since by definition, 
    $\mu_{r_i'} > \mu_{r_i'+1}$ and $\mu_{r_i} > \mu_{r_i+1}$, we conclude 
    that the isle of $1$s shifts over time to the right, and because the 
    configuration is a ring, this gives rise to a cycle.
    
  \item $\mu_{p-1}<\mu_p$ and $\mu_q<\mu_{q+1}$. 
    This case is similar to Case~3, except that the isle of $1$s shifts over 
    time to the left.
  \end{enumerate}\medskip

  Now, let us denote by $\{r_i\}_{i=1}^N$, with $N \in \integers{n}$, the set 
  of cells such that $\mu_{r_i} > \mu_{r_i+1}$). 
  Using these $r_i$, we can partition the ring into sections from cell $r_i+1$ 
  to cell $r_{i+1}$, as shown on the following table:
  \begin{center}
    \begin{tabular}{cc|ccc|ccc|ccccccccccc}
      \hline
      $\dots$ & $r_{i-1}$ & $r_{i-1}+1$ & $\dots$ & $r_i$ & $r_i+1$ & $\dots$ 
      & $r_{i+1}$ & $r_{i+1}+1$ & $\dots$\\
      \hline
    \end{tabular}
  \end{center}
  If the isle of $1$s starts on a cell $p \not \in \{r_l\}_{l=1}^N$ for all $i 
  \in \{1,\dots,N\}$, and ends on a cell $q \in \{r_l\}_{l=1}^N$, then we are 
  in Case~1, and this situation can only lead to a homogeneous fixed point.
  If the isle starts on a cell $p \in \{r_l\}_{l=1}^N$, and ends on a cell $q 
  \not \in \{r_l\}_{l=1}^N$, then we are in Case~2, which also leads to a 
  homogeneous fixed point.
  If there exists $r_i, r_j \in \{r_l\}_{l=1}^N$ such that $p = r_i+1$ and $q 
  = r_j$, with $1 \leq i < j \leq N$, then we are on Case~3, and it leads to a 
  cycle. 
  As seen in Case~3, we know that we gain $1$s to the right until we reach 
  from cell $r_j$ to cell $r_{j+1}$. 
  Similarly, we know that we lose $1$s from cell $r_i+1$ to cell $r_{i+1}+1$, 
  as shown on the following graphic.
  \begin{center}
    \begin{tabular}{cc|ccc|ccc|ccc|cccccccc}
      \hline
      $\dots$ & $r_{i}$ & $r_{i}+1$ & $\dots$ & $r_{i+1}$ & $r_{i+1}+1$ & 
      $\dots$ & $r_{j}$ & $r_{j}+1$ & $\dots$ & $r_{j+1}$ & $r_{j+1}+1$ & 
      $\dots$\\
      \hline\hline
      $\dots$ & $0$ & $1$ & $\dots$ & $1$ & $1$ & $\dots$ & $1$ & $0$ & 
      $\dots$ & $0$ & $0$ & $\dots$\\
      \hline
      $\dots$ & $0$ & $0$ & $\dots$ & $0$ & $1$ & $\dots$ & $1$ & $1$ & 
      $\dots$ & $1$ & $0$ & $\dots$\\
      \hline
    \end{tabular}
  \end{center}
  It is not difficult to see that after $N-j$ iterations, the isle of $1$s 
  will have to start on $r_{i*}+1$, with $i* = i+N-j$, and to end on $r_N$ 
  ($N=j+(N-j)$). 
  Since we are working on a ring, the section after $r_N$ goes from $r_N+1$ to 
  $r_1$, meaning that at the $N-j+1$ iteration, the isle of $1$s starts on 
  $r_{i*+1}$ and ends on $r_1$.
  $j-1$ iterations after that, the isle will go from $r_i+1$ to $r_j$, 
  completing the cycle ($i = i+(N-j)+1+(j-1)$ and $j=j+(N-j)+1+(j-1)$), 
  meaning that it takes $(N-j)+1+(j-1)=N$ iterations to complete the 
  cycle.\medskip
  
  Observe that, until here, we have been working with a single isle of $1$s. 
  Let us assume now that the configuration has two isles of $1$s and let us 
  denote them by $B_1$ and $B_2$).
  Consider the following proof based on a case disjunction:
  \begin{enumerate}
  \item Let us assume that there are $i, j, k \in \{1, \dots, N\}$ such that 
    $B_1$ starts on $r_i+1$ and ends on $r_j$, while $B_2$ starts on $r_k+1$ 
    and ends on $q \not \in \{r_l\}_{l=1}^N$. 
    Because there are two isles, we know that $1 \leq i < j < k \leq N$.
    \begin{center}
      \begin{tabular}{cc|ccc|ccc|ccc|cccccccc}
        \hline
        $\dots$ & $r_{i}$ & $r_{i}+1$ & $\dots$ & $r_{j}$ & $r_{j}+1$ & 
        $\dots$ & $r_{k}$ & $r_{k}+1$ & $\dots$ & $q$ & $q+1$ & $\dots$\\
        \hline\hline
        $\dots$ & $0$ & $1$ & $\dots$ & $1$ & $0$ & $\dots$ & $0$ & $1$ & 
        $\dots$ & $1$ & $0$ & $\dots$\\
        \hline
      \end{tabular}
    \end{center}
    Since with each iteration $B_1$ can only gain $1$s up to the next cell 
    such that $\mu_r > \mu_{r +1}$, this means that at most 
    \begin{equation*}
      r_{j+1} = r_k < r_{k+1} \implies r_{j+t} = r_{k+t-1} < r_{k+t}\text{.}
    \end{equation*}
    This means that the two isles of $1$s stay separated.
    Moreover, since $q \not \in \{r_l\}_{l=1}^N$ then $B_2$ will also lose 
    $1$s to the right. 
    This means that $B_2$ will disappear without interacting with $B_1$.\\
    The case with $i, j, k \in \{1, \dots, N\}$ such that $B_2$ starts on 
    $r_i+1$ and ends on $r_j$, while $B_1$ starts on $r_k+1$ but ends on 
    $q \not \in \{r_l\}_{l=1}^N$ is similar.
    \begin{center}
      \begin{tabular}{cc|ccc|ccc|ccc|cccccccc}
        \hline
        $\dots$ & $r_k$ & $r_k+1$ & $\dots$ & $q$ & $q+1$ & $\dots$ & $r_{i}$ 
        & $r_{i}+1$ & $\dots$ & $r_{j}$ & $r_{j}+1$ & $\dots$\\
        \hline\hline
        $\dots$ & $0$ & $1$ & $\dots$ & $1$ & $0$ & $\dots$ & $0$ & $1$ & 
        $\dots$ & $1$ & $0$ & $\dots$\\
        \hline
      \end{tabular}
    \end{center}
    
  \item Let us assume that there are $i, j, k \in \{1, \dots, N\}$ such that 
    $B_1$ starts on $r_i+1$ and ends on $r_j$, while $B_2$ ends on $r_k$ but 
    starts on $p \not \in \{r_l\}_{l=1}^N$. 
    Because there are two isles, we know that $r_j<p$.
    \begin{center}
      \begin{tabular}{cc|ccc|ccc|ccc|cccccccc}
        \hline
        $\dots$ & $r_{i}$ & $r_{i}+1$ & $\dots$ & $r_{j}$ & $r_{j}+1$ & 
        $\dots$ & $p-1$ & $p$ & $\dots$ & $r_k$ & $r_k+1$ & $\dots$\\
        \hline\hline
        $\dots$ & $0$ & $1$ & $\dots$ & $1$ & $0$ & $\dots$ & $0$ & $1$ & 
        $\dots$ & $1$ & $0$ & $\dots$\\\hline
      \end{tabular}
    \end{center}
    Since cell $r_i$ is updated \emph{after} cell $r_i+1$, then when cell 
    $r_i+1$ updates, its neighborhood will be $(01*)$ and $f_{178}(01*) = 0$, 
    with $* \in \{0,1\}$, and:
    \begin{itemize}
    \item $B_1$ will lose $1$s to the left up to $r_{i+1}(\leq r_j)$,
    \item $B_1$ will gain $1$s to the right up to $r_{j+1}$ (which could 
      eventually be greater than $p$),
    \item $B_2$ will gain $1$s to the left until it reaches $B_1$, and 
    \item $B_2$ will gain $1$s to the right up to $r_{k+1}(\leq r_i)$.
    \end{itemize}
    Thus, the isles $B_1$ and $B_2$ end up forming a single isle $B'$ which 
    starts on $r_{i+t}+1$ and ends on $r_{k+t}$, after $t$ iterations.
    This is analogous to the case with $i, j, k \in \{1, \dots, N\}$ such that 
    $B_2$ starts on $r_i+1$ and ends on $r_j$, while $B_1$ ends on $r_k+1$ and 
    starts on $p \not \in \{r_l\}_{l=1}^N$.
    
  \item Let us assume that there are $i, j \in \{1, \dots, N\}$ such that 
    $B_1$ starts on $r_i+1$ and ends on $r_j$, while $B_2$ starts on $p+1$ and 
    ends on $q$, with $p, q \not \in \{r_l\}_{l=1}^N$. 
    We know that $r_j < p$. 
    \begin{center}
      \begin{tabular}{cc|ccc|ccc|ccc|cccccccc}
        \hline
        $\dots$ & $r_{i}$ & $r_{i}+1$ & $\dots$ & $r_{j}$ & $r_{j}+1$ & 
        $\dots$ & $p-1$ & $p$ & $\dots$ & $q$ & $q+1$ & $\dots$\\
        \hline\hline
        $\dots$ & $0$ & $1$ & $\dots$ & $1$ & $0$ & $\dots$ & $0$ & $1$ & 
        $\dots$ & $1$ & $0$ & $\dots$\\
        \hline
    \end{tabular}
  \end{center}
  Similarly to the previous case, $B_1$ join $B_2$ after $t$ iterations, when 
  $r_{j+t} > p'$. 
  The resulting isle $B'$ loses its $1$s to the left and to the right, meaning 
  that the dynamics leads to a fixed point.
  
  \item Let us assume that there are $i, j, k, l \in \{1,\dots,N\}$ such that 
    $B_1$ starts on $r_i+1$ and ends on $r_j$, while $B_2$ starts on $r_k+1$ 
    and ends on $r_l$ ($1 \leq i< j < k < l \leq N$). 
    \begin{center}
      \begin{tabular}{cc|ccc|ccc|ccc|cccccccc}
        \hline
        $\dots$ & $r_{i}$ & $r_{i}+1$ & $\dots$ & $r_{j}$ & $r_{j}+1$ & 
        $\dots$ & $r_k$ & $r_k+1$ & $\dots$ & $r_l$ & $r_l+1$ & $\dots$\\
        \hline\hline
        $\dots$ & $0$ & $1$ & $\dots$ & $1$ & $0$ & $\dots$ & $0$ & $1$ & 
        $\dots$ & $1$ & $0$ & $\dots$\\
        \hline
      \end{tabular}
    \end{center}
    As previously shown, $B_1$ cannot reach $B_2$ ($r_{j+t} < r_{k+t}$ because 
    $r_j < r_k$). 
    Similarly, $B_2$ cannot reach $B_1$ because if there exists an iteration 
    $t$ for which $r_{l+t} = r_{i+t}$, then $r_l = r_i$ and there could not 
    have been two separate isles to begin with. Thus, $B_1$ and $B_2$ do not 
    interact with each other.
  \end{enumerate}

  Finally, let there be a configuration $w \in \{0,1\}^n$. 
  This configuration can be written as a set of isles of $1$s denoted by 
  $B_i$, with $i \in \{1,\dots,N\}$ with $0$s in-between, such that:
  \begin{equation*}
    w = \dots 0 B_1 0 \dots 0 B_{i-1} 0 \dots 0 B_i 0 \dots 0 B_{i+1} 0 \dots 
    0 B_N0 \dots\text{.}
  \end{equation*}
  We can find four kinds of isles, which correspond to the cases studied in 
  the first step of this proof.
  
  \begin{itemize}
  \item Notice that all isles corresponding to Case~2, denoted by $B_{><}$ 
    (meaning $\mu_{p-1} > \mu_p$ and $\mu_q < \mu_{q+1}$) disappear without 
    interacting with other isles, which is why we will only consider:
    \begin{compactitem}
    \item isles of type~1 corresponding to Case~1, denoted by $B_{<>}$ 
      (meaning $\mu_{p-1} < \mu_p$ and $\mu_q > \mu_{q+1}$), 
    \item isles of type~3 corresponding to Case~3, denoted by $B_{>>}$ 
      (meaning $\mu_{p-1} > \mu_p$ and $\mu_q > \mu_{q+1}$), and 
    \item isles of type~4 corresponding to Case~4, denoted by $B_{<<}$ 
      (meaning $\mu_{p-1} < \mu_p$ and $\mu_q<\mu_{q+1}$).
    \end{compactitem}
  \item It is easy to see that if every isle is $B_{<>}$, it leads to a fixed 
    point; and if every isle is of type $B_{>>}$ or of type $B_{<<}$, then the 
    isles do not interact with each other, leading to cycles of length less 
    than $n$.
  \item If there is $B_i$ of type~3 such that $B_{i+1}$ is of type~4, then we 
    have shown that both isles disappear. 
    Thus, $B_{>>}$ and $B_{<<}$ cancel out.
  \item When an isle of type $B_{<>}$ reaches another isle, they fuse into a 
    single isle which will be of the same type as the one that was reached. 
    Note that if there is a section of the configuration such that 
    \begin{equation*}
      \dots 0 B_{>>} 0^* B_{<>} 0^* B_{<<}0\dots\text{,}
    \end{equation*}
    it does not matter which isle $B_{<>}$ reaches first, the resulting isle 
    will cancel the space.
  \item Thus, we know that the number of isles $B_{>>}$ decreases if and only 
    if the number of $B_{<<}$ also decreases.
  \item Therefore:
    \begin{compactitem}
    \item If $|B_{>>}| = |B_{<<}|$, then the configuration reaches a fixed 
      point.
    \item If $|B_{>>}|>|B_{<<}|$ or $|B_{>>}|<|B_{<<}|$, then we will reach a 
      cycle of length strictly less than $n$.\hfill\qed
    \end{compactitem}
  \end{itemize}
\end{proof}

In order to tackle the asymptotical dynamics of the family $(178,\bs)$, we 
make use of a the following lemma which shows that a subclass of block-
sequential dynamics of ECA~$178$ can be simulated thanks to a sequential 
update mode.

\begin{lemma}\label{lem:equi_seq_bs*}
  For all $\mu_a \in \bs, \mu = (B_0, \dots, B_{p-1})$ such that for all $i 
  \in \integers{n}$, if $i \in B_k$ then $i-1, i+1 \not \in B_k$ (with $k \in 
  \integers{p}$, there is an update mode $\mu_b \in \seq$ such that for all $x 
  \in \B^n$
  \begin{equation*}
    F_a(x) = F_b(x), \text{ with }
    F_a = (f, \mu_a),
    F_b = (f, \mu_b) \text{ and }
    f \text{ an ECA.}
  \end{equation*}
\end{lemma}

\begin{proof}
  Let $B_k = \{b_0^{B_k}, b_1^{B_k}, \dots, B_{j_k}^{B_k}\}$, with $k \in 
  \integers{p}$.
  Because of the hypothesis, $i \in \integers{n}$ cannot be updated at the 
  same substep as either of its neighbors. 
  In other words, for all $k \in \integers{p}$, this means that automata of 
  block $B_k$ do not interact with each other and thus, that $B_k$ can be 
  subdivided into as many subsets as its cardinal, each subset being composed 
  of one automaton of $B_k$. 
  From there, we can define
  \begin{equation*}
    \mu_b = \left(
      b_0^{B_0}, \dots, b_{j_0}^{B_0}, \dots, b_1^{B_k}, \dots, b_{j_k}^{B_k}, 
      \dots, b_1^{B_{p-1}}, \dots, b_{j_{p-1}}^{B_{p-1}}
    \right)\text{.}
  \end{equation*}
  By definition, after the first substep of $\mu_a$, only the automata 
  belonging to block $B_0$ are updated, and after the $j_0$th substep, the 
  same cells will have been updated for $\mu_b$. 
  Since none of the cells belonging to $B_0$ are neighbors, the fact that they 
  are updated one after the other instead of all of them at the same substep 
  does not change the result.\smallskip
  
  Recursively, we can see that with $\mu_a$, for all $k \in \integers{p}$, all 
  cells belonging to the sets $B_0, \dots, B_{k-1}$ have been updated at the 
  $k$th substep, while at the $\hat{j_k}$th substep, the same automata will 
  have been updated with $\mu_b$, with $\hat{j_k} = \sum_{\ell=0}^{k-1} 
  j_{\ell}$.
  This leads to the conclusion that for all $x \in \B^n,\ F_a(x) = 
  F_b(x)$.\hfill\qed
\end{proof}

\begin{thm}\label{thm:178BS}
  $(178,\bs)$  applied over a grid of size $n$  has largest limit cycles of 
  length $O(n)$.
\end{thm}

\begin{proof}
    We will divide the block-sequential update modes in two groups: 
    \begin{enumerate}
    \item[A)] those where there are at least two consecutive cells are updated 
      simultaneously, and
    \item[B)] those where there are not.
    \end{enumerate}
     More formally, considering $\bs{}$ modes of period $p$ with blocks $B_0, 
     \dots, B-{p-1}$, we distinguish the set
     \begin{itemize}
     \item $\mathcal{A} = \{\mu \in \BS: k \in \integers{p}, forall i \in 
       \integers{n-1}, \ i \in B_k \implies i-1, i+1 \notin B_k\}$, and 
     \item $\mathcal{B} = \BS{} \setminus \mathcal{A} = \{\mu \in \BS: k \in 
       \integers{p}, \exists i \in \integers{n},\ i \in B_k \implies i+1 \in 
       B_k\}$.
     \end{itemize}

    From Lemma~\ref{lem:equi_seq_bs*}, we know that each update mode of 
    $\mathcal{A}$ can be written as a sequential update mode, and from 
    Theorem~\ref{thm:178SEQ}, we already know that family $(178, \mathcal{A})$ 
    has largest limit cycles of length $O(n)$.\smallskip
    
    Thus, we will turn our attention to the dynamics of ECA $178$ induced by 
    $\mathcal{B}$, composed of update modes in which there are at least two 
    consecutive cells that belong to the same block.
    Let $\{s_i\}_{i=1}^{N_s}$ be the set of cells that are updated at the same 
    time as their right-side neighbor ($\mu_{s_i} = \mu_{s_{i}+1}$).
    First, let us assume that only one cell $s$ belongs to 
    $\{s_i\}_{i=1}^{N_s}$ ($N_s = 1$) and that the initial configuration is 
    $x = 10^{n-1}$, where only cell $s$ is at state $1$, such that:
    \begin{center}
      \begin{tabular}{c|c|c|c||c|c||c|c|c|c}
        \hline
        $\dots$ & & $\dots$ & $s-1$ & $s$ & $s+1$ & $s+2$ & $\dots$ & & 
        $\dots$\\
        \hline\hline
        $\dots$ & $0$ & $\dots$ & $0$ & $1$ & $0$ & $0$ & $\dots$ & $0$ & 
        $\dots$\\
        \hline
    \end{tabular}
  \end{center}
   From the definition of the rule, we know that after the first iteration, 
   cell $s$ must become $0$, while $s+1$ must be $1$; to know what happens on 
   the rest of the configuration, similarly to the proof of 
   Theorem~\ref{thm:178SEQ}, we must proceed with a case 
   disjunction.
   
  \begin{itemize}
  \item $\mu_{s-1}<\mu_s$ and $\mu_{s+1}>\mu_{s+2}$.
    From previous analysis, we know that we gain $1$s to the left until a cell 
    $\ell$ such that $\mu_{\ell_1} < \mu_\ell$ and $\mu_\ell < \mu_{\ell+1}$, 
    and we gain $1$s to the right until a cell $r$ such that $\mu_{r-1} < 
    \mu_r$ and $\mu_r > \mu_{r+1}$, which in this case happens to be $s+1$.
    This means that after the first iteration, we will have two new isles of 
    $1$s, each moving in opposite directions. 
    We know that once the isles of $1$s circumnavigate the ring they will 
    eventually meet and, since they're moving in opposite directions, they 
    will cancel each other out.
    \begin{center}
      \begin{tabular}{cc|ccc|ccc||c|c||ccc|cccc|cc}
        \hline
        $\dots$ & $\ell'-1$ & $\ell'$ & $\dots$ & $\ell-1$ & $\ell$ & $\dots$ 
        & $s-1$ &$s$ & $s+1$ & $s+2$ & $\dots$ & $r'$ & $r'+1$ & $\dots$\\
        \hline\hline
        $\dots$ & $0$ & $0$ & $\dots$ & $0$ & $0$ & $\dots$ & $0$ & $1$ & $0$ 
        & $0$ & $\dots$ & $0$ & $0$ & $\dots$\\
        \hline
        $\dots$ & $0$ & $0$ & $\dots$ & $0$ & $1$ & $\dots$ & $1$ & $0$ & $1$ 
        & $0$ & $\dots$ & $0$ & $0$ & $\dots$\\
        \hline
      \end{tabular}
    \end{center}
    Note that, because $\mu_s = \mu_{s+1}$, the state of cell $s$ returns to 
    $1$ at the next iteration, as well as that of $s+1$ will returns to $0$. 
    Since cells $s$ and $s+1$ will be $10$ and $01$ alternatively, the 
    dynamics will have a limit cycle of length $2$.
    \begin{center}
      \begin{tabular}{cc|ccc|ccc||c|c||ccc|cccc|ccc}
        \hline
        $\dots$ & $\ell'-1$ & $\ell'$ & $\dots$ & $\ell-1$ & $\ell$ & $\dots$ 
        & $s-1$ & $s$ & $s+1$ & $s+2$ & $\dots$ & $r'$ & $r'+1$ & $\dots$\\
        \hline\hline
        $\dots$ & $0$ & $0$ & $\dots$ & $0$ & $0$ & $\dots$ & $0$ & $1$ & $0$ 
        & $0$ & $\dots$ & $0$ & $0$ & $\dots$\\
        \hline
        $\dots$ & $0$ & $0$ & $\dots$ & $0$ & $1$ & $\dots$ & $1$ & $0$ & $1$ 
        & $0$ & $\dots$ & $0$ & $0$ & $\dots$\\
        \hline
        $\dots$ & $0$ & $1$ & $\dots$ & $1$ & $0$ & $\dots$ & $0$ & $1$ & $0$ 
        & $1$ & $\dots$ & $1$ & $0$ & $\dots$\\
        \hline
        $\dots$ & $1$ & $0$ & $\dots$ & $0$ & $1$ & $\dots$ & $1$ & $0$ & $1$ 
        & $0$ & $\dots$ & $0$ & $1$ & $\dots$\\
      \hline
      \end{tabular}
    \end{center}
    
  \item $\mu_{s-1} > \mu_s$ and $\mu_{s+1} < \mu_{s+2}$. 
    Similarly to what we have already established, since $\mu_{s+1} < 
    \mu_{s+2}$, we know that we have to gain $1$s to the right until a cell 
    $r$ (which cannot be $s+1$), meaning that the new isle of $1$s will go 
    from $s+1$ to $r$. 
    Moreover, because after the next iteration of the rule, the state of $s$ 
    will once again be $1$, there will also be an isle of $1$s moving to the 
    left, which will go from $\ell$ to $s$.
    \begin{center}
      \begin{tabular}{cc|ccc|cccc||cccc|cccc|ccc}
        \hline
        $\dots$ & $\ell'-1$ & $\ell'$ & $\dots$ & $\ell-1$ & $\ell$ & $\dots$ 
        & $s-1$ & $s$ & $s+1$ & $s+2$ & $\dots$ & $r$ & $r+1$ & $\dots$\\
        \hline\hline
        $\dots$ & $0$ & $0$ & $\dots$ & $0$ & $0$ & $\dots$ & $0$ & $1$ & $0$ 
        & $0$ & $\dots$ & $0$ & $0$ & $\dots$\\
        \hline
        $\dots$ & $0$ & $0$ & $\dots$ & $0$ & $0$ & $\dots$ & $0$ & $0$ & $1$ 
        & $1$ & $\dots$ & $1$ & $0$ & $\dots$\\
        \hline
        $\dots$ & $0$ & $0$ & $\dots$ & $0$ & $1$ & $\dots$ & $1$ & $1$ & $0$ 
        & $0$ & $\dots$ & $0$ & $1$ & $\dots$\\
        \hline
        $\dots$ & $0$ & $1$ & $\dots$ & $1$ & $0$ & $\dots$ & $0$ & $0$ & $1$ 
        & $1$ & $\dots$ & $1$ & $0$ & $\dots$\\
        \hline
      \end{tabular}
    \end{center}
    Just like in the previous case, we now have two isles of $1$s moving in 
    opposite directions, destined to cancel each other out, while at the same 
    time $s$ and $s+1$ are locked in a limit cycle of length $2$.
    
  \item $\mu_{s-1} > \mu_s$ and $\mu_{s+1} > \mu_{s+2}$. 
    This case is a combination of the previous two. 
    To the right, we must repeat the analysis of Case~1, and to the left we 
    must repeat the analysis of Case~2.
    This means that once again we have a dynamics that ends into a limit cycle 
    of length $2$.
    \begin{center}
      \begin{tabular}{cc|ccc|cccc||c||ccc|cccc|ccc}
        \hline
        $\dots$ & $\ell'-1$ & $\ell'$ & $\dots$ & $\ell-1$ & $\ell$ & $\dots$ 
        & $s-1$ & $s$ & $s+1$ & $s+2$ & $\dots$ & $r$ & $r+1$ & $\dots$\\
        \hline\hline
        $\dots$ & $0$ & $0$ & $\dots$ & $0$ & $0$ & $\dots$ & $0$ & $1$ & $0$ 
        & $0$ & $\dots$ & $0$ & $0$ & $\dots$\\
        \hline
        $\dots$ & $0$ & $0$ & $\dots$ & $0$ & $0$ & $\dots$ & $0$ & $0$ & $1$ 
        & $0$ & $\dots$ & $0$ & $0$ & $\dots$\\
        \hline
        $\dots$ & $0$ & $0$ & $\dots$ & $0$ & $1$ & $\dots$ & $1$ & $1$ & $0$ 
        & $1$ & $\dots$ & $1$ & $0$ & $\dots$\\
        \hline
        $\dots$ & $0$ & $1$ & $\dots$ & $1$ & $0$ & $\dots$ & $0$ & $0$ & $1$ 
        & $0$ & $\dots$ & $0$ & $1$ & $\dots$\\
        \hline
      \end{tabular}
    \end{center}
    
  \item Case 4: $\mu_{s-1}<\mu_s$ and $\mu_{s+1}<\mu_{s+2}$. 
    This case is symmetric to Case~3.
    \begin{center}
      \begin{tabular}{cc|ccc|ccc|c|cccc|cccc|ccc}
        \hline
        $\dots$ & $\ell'-1$ & $\ell'$ & $\dots$ & $\ell-1$ & $\ell$ & $\dots$ 
        & $s-1$ & $s$ & $s+1$ & $s+2$ & $\dots$ & $r$ & $r+1$ & $\dots$\\
        \hline\hline
        $\dots$ & $0$ & $0$ & $\dots$ & $0$ & $0$ & $\dots$ & $0$ & $1$ & $0$ 
        & $0$ & $\dots$ & $0$ & $0$ & $\dots$\\    
        \hline
        $\dots$ & $0$ & $0$ & $\dots$ & $0$ & $1$ & $\dots$ & $1$ & $0$ & $1$ 
        & $1$ & $\dots$ & $1$ & $0$ & $\dots$\\
        \hline
        $\dots$ & $0$ & $1$ & $\dots$ & $1$ & $0$ & $\dots$ & $0$ & $1$ & $0$ 
        & $0$ & $\dots$ & $0$ & $1$ & $\dots$\\  
        \hline
        $\dots$ & $1$ & $0$ & $\dots$ & $0$ & $1$ & $\dots$ & $1$ & $0$ & $1$ 
        & $1$ & $\dots$ & $1$ & $0$ & $\dots$\\
        \hline
      \end{tabular}
    \end{center}
  \end{itemize}

  Morevover, if there is an isle of $1$s moving from the left to the right (or 
  from the right to the left), as soon as it reaches cell $s$, the 
  configuration will turn to the one that we have already analyzed.
  \begin{center}
    \begin{tabular}{cc|ccc|c|cccccc}
      \hline
      $\dots$ & $r_j$ & $r_j+1$ & $\dots$ & $r_i$ & $s$ & $s+1$ & $s+2$ & 
      $\dots$ & $\dots$\\
      \hline\hline
      $\dots$ & $0$ & $1$ & $\dots$ & $1$ & $0$ & $0$ & $0$ & $\dots$ & $0$ & 
      $\dots$\\
      \hline
      $\dots$ & $0$ & $0$ & $\dots$ & $0$ & $1$ & $0$ & $0$ & $\dots$ & $0$ & 
      $\dots$\\
      \hline
    \end{tabular}
  \end{center}
  Furthermore, this continues to hold when there is a set of cells 
  $\{s_i\}_{i=1}^{N}$ such that $\mu_s = \mu_{s+1}$, because we know that 
  isles of $1$s which travel in opposite directions will cancel each other 
  out.
  This means that for every block-sequential update mode such that there are 
  (at least) two consecutive cells that update on the same sub-step, the 
  longest possible limit cycle is of length $2$.\hfill\qed
\end{proof}

\begin{thm}\label{thm:178PAR}
  $(178,PAR)$ has largest limit cycles of length 2.
\end{thm}

\begin{proof}
  Direct from the proof of Theorem~\ref{thm:178BS}.\hfill\qed
\end{proof}

%%%%%%%%%%%%%%%%%%%%%%%%%%
\section{Discussion}
\label{sec:persp}

In this paper, we have focused in the study of two ECA rules under different 
update modes. 
These rules are the rule $156$ and the rule $178$. 
These rule selections were made subsequent to conducting numerical simulations 
encompassing a set of $88$ non-equivalent ECA rules, each subjected to diverse 
update modes. 
Rule $156$ and Rule $178$ emerged as pertinent subjects for further 
investigation due to their pronounced sensitivity to asynchronism, as 
substantiated by our simulation results.
Following this insight we have analytically shown two different behaviors 
illustrated in Table~\ref{table:rules_and_updatemodes}:
\begin{itemize}
\item Rule 156: the maximum period of the attractors changes from constant to 
  superpolynomial when a bipartite (block sequential with two blocks) update 
  modes are considered.
\item Rule 178: the same phenomenon is observed but "gradually increasing", 
  from constant, to linear and then to superpolynomial.
\end{itemize}
The obtained results suggest that it might not be a unified classification 
according to our measure of complexity (the maximum period of the attractors). 
This observation presents an open question regarding what other complexity 
measures can be proposed to classify ECA rules under different update 
modes.\smallskip

By analyzing our simulations we have found interesting observations about the 
dynamical complexity of ECA rules. 
An example that we identified in the simulations but is not studied in the 
paper is the one of the rule 184 (known also as the traffic rule). 
In this case we have shown that under the bipartite update mode there are 
only fixed points. 
This is interesting considering the fact that it is known that the maximum 
period can be linear in the size of the network for the parallel update 
mode~\cite{li1992phenomenology}. 
Thus, in this case, the rule seems to exhibit a different kind of dynamical 
behaviour compared to rule $156$ and $178$ (asynchronism tend to produce 
simpler dynamics instead of increasing the complexity) and might be 
interesting to study from theoretical standpoint. 
In addition, it could be interesting to study if other rules present this 
particular behaviour.\smallskip

Finally, from what we have been able to observe on our simulations, there 
exists a fourth class of rules where the length of the longest limit cycles 
remains constant regardless of the update mode (for example rule $4$, $8$, 
$12$, $72$, and $76$ exhibit only fixed points when tested exhaustively for 
each configuration of size at most $n=8$). 
Thus, this evidence might suggest the existence of rules that are robust with 
respect to asynchronism. 
In this sense, it could be interesting to analytically study some of these 
rules and try to determine which dynamical property makes them robust under 
asynchronism.

%%%%%%%%%%%%%%%%%%%%%%%%%%
\paragraph{Acknowledgments} 

The authors are thankful to 
	projects ANR-18-CE40-0002 \linebreak{} ``FANs'' (MRW, SS),
	Fondecyt-ANID 1200006 (EG)
	MSCA-SE-101131549 \linebreak{} ``ACANCOS'' (EG, IDL, MRW, SS), 
	STIC AmSud 22-STIC-02 (EG, IDL, MRW, SS), and 
	for their funding.

%%%%%%%%%%%%%%%%%%%%%%%%%%
%\bibliographystyle{plain}
%\bibliography{arxiv_latin2024.bib}

\end{document}